\documentclass[12pt,a4paper,draftcls,onecolumn]{IEEEtran}
\date{}
\usepackage{amsfonts}      
     \usepackage{amsmath}

\usepackage{epsfig}
\usepackage{pmat}
\usepackage{dsfont}
\usepackage{mathrsfs}
\usepackage{amsfonts}      

\usepackage{graphics} 
\usepackage{epsfig}
\usepackage{pmat}
\usepackage{dsfont}
\usepackage{pmat}
\usepackage{graphics}
\usepackage{epsfig}
\usepackage{dsfont}
\usepackage{pmat}
\usepackage{multirow}

\usepackage{mathrsfs}
\usepackage{makecell}
\usepackage{booktabs}
\usepackage{setspace}

\usepackage{amssymb}
\usepackage{color}

\usepackage{epsfig}

\newtheorem{theorem}{Theorem}[section]
\newtheorem{remark}[theorem]{Remark}
\newtheorem{lemma}[theorem]{Lemma}
\newtheorem{proposition}[theorem]{Proposition}
\newtheorem{definition}[theorem]{Definition}
\newtheorem{example}[theorem]{Example}

\usepackage{algorithm}
\usepackage{algorithmic}

\usepackage[top=0.5in,bottom=0.5in,left=0.9in,
right=0.7in,dvips,dvipdfm,includehead,includefoot
            ]{geometry}
\title{Balanced Truncation of Linear Time-Invariant Systems over Finite-frequency Ranges}

\author{Xin~Du$^{1,2}$,~\IEEEmembership{}
        Peter~Benner$^{1,2*}$,~\IEEEmembership{}
\thanks{$^1$ Max Planck Institute for
Dynamics of Complex Technical Systems, Sandtorstra$\ss$e 1, 39106
Magdeburg, Germany. \emph{$^{*}$Corresponding author:}\texttt{
benner@mpi-magdeburg.mpg.de}}
\thanks{$^2$School of Mechatronic Engineering and Automation,
      Shanghai University,  Shanghai, 200072, P.~R.~China. }
       \thanks{This work was supported by by NSFC under Grant (61304143) and the High-End Foreign Expert Program of the P.~R.~China (GDT20153100033).}}
\begin{document}
\centerfigcaptionstrue

\maketitle 

\begin{abstract}
This paper discusses model order reduction of LTI systems over limited frequency intervals within the framework of balanced truncation. Two new \emph{frequency-dependent balanced truncation} methods were developed, one is \emph{SF-type frequency-dependent balanced truncation} to copy with the cases that only a single dominating point of the operating frequency interval is pre-known, the other is \emph{interval-type frequency-dependent balanced truncation} to deal with the cases that both of the upper and lower bound of frequency interval are known \emph{a priori}. SF-type error bound and interval-type error bound are derived for the first time to estimate the desired approximation error over pre-specified frequency interval. We show that the new methods generally lead to good in-band approximation performance, at the same time, provide accurate error bounds under certain conditions. Examples are included for illustration.\end{abstract}



%

\section{Introduction and Problem Formulations}

\noindent We study model order reduction for linear time-invariant continuous-time systems
\begin{equation}
\label{originalsystem}
\small
G(\jmath\omega ):\left\{\begin{array}{l}
{\dot x(t) = Ax(t) + Bu(t)}  \\
   {y(t) = Cx(t) + Du(t)}  \\
   \end{array}\right.
 \Leftrightarrow G(\jmath\omega ): = \pmatset{1}{0.36pt}
  \pmatset{0}{0.2pt}
  \pmatset{2}{4pt}
  \pmatset{3}{4pt}
  \pmatset{4}{4pt}
  \pmatset{5}{4pt}
  \pmatset{6}{4pt}    \begin{pmat}[{|}]
       A   &   B \cr\-
       C   &   D \cr
      \end{pmat}   \Leftrightarrow G(\jmath\omega ):\mathop  = \limits^\omega  C{(\jmath\omega I - A)^{ - 1}}B + D\end{equation}
\noindent where $A\in \mathbb C^{n\times n},  B \in \mathbb
C^{n\times m},  C \in \mathbb C^{p\times n},  D \in \mathbb C^{p\times
m}$, $x(t) \in \mathbb C^n$ is the state vector, $u(t)\in \mathbb
C^m$ is the input signal, $y(t)\in \mathbb C^p$ is the output
signal. Modeling of complex physical processes often leads to large order $n$. The corresponding high storage requirements and expensive computations make it very difficult to simulate, optimize or even design such large scale systems \cite{Benner}-\cite{Autolas2}. In this case model order
reduction (MOR) plays an important role. It consists in approximating the system (1) by a reduced-order system: \begin{equation}
\label{reducedmodel}
\small
G_r(\jmath\omega ):\left\{\begin{array}{l}
{\dot x_r(t) = A_rx_r(t) + B_ru(t)}  \\
   {y(t) = C_rx_r(t) + D_ru(t)}  \\
   \end{array}\right.
 \Leftrightarrow G_r(\jmath\omega ): = \pmatset{1}{0.36pt}
  \pmatset{0}{0.2pt}
  \pmatset{2}{4pt}
  \pmatset{3}{4pt}
  \pmatset{4}{4pt}
  \pmatset{5}{4pt}
  \pmatset{6}{4pt}    \begin{pmat}[{|}]
       A_r   &   B_r \cr\-
       C_r   &   D_r \cr
      \end{pmat}   \Leftrightarrow G_r(\jmath\omega ):\mathop  = \limits^\omega  C_r{(\jmath\omega I - A_r)^{ - 1}}B_r + D_r\end{equation}
\noindent where $A_r\in \mathbb C^{r\times r}, B_r \in \mathbb
C^{n\times m}, C_r \in \mathbb C^{p\times n}, D_r \in \mathbb C^{p\times
m}$ with $r < n$.

Balanced truncation is a well grounded and the most commonly used model order reduction scheme \cite{BT2_Survey} \cite{BT3_PRBT}. The standard form is the so-called \emph{Lyapunov balanced truncation}, which was first introduced in the systems and control
literature by Moore \cite{BT1_Moore}. The prominent advantages of balanced truncation is that it preserves stability and provides an \emph{a priori} known error bound over the entire-frequency range. In detail, it gives a upper bound of the following entire-frequency (EF) type approximation performance index function
\begin{equation}
      \label{EF_index}
       \sigma_{max}({G(\jmath\omega ) - {G_r}(\jmath\omega )}) ,  {\kern 15pt} for {\kern 4pt} all {\kern 4pt}  \omega \in [-\infty, +\infty]
      \end{equation}
In many practical applications, the operating frequency of input signal belongs to a fully or partially known finite-frequency range such as a limited interval (i.e. $\omega \in [\varpi_1,\varpi_2]$). For
those cases, the reduced-order model is only needed to capture the input-output behavior of the original system for input signals with admissible frequency. Correspondingly, good in-band approximation performance is more expected,  while the out-band approximation performance might be neglected \cite{SPA_benner}-\cite{FGBT_Ghafoor2008}. In other words, the objective of \emph{finite-frequency} (FF) model order reduction is only to minimize the following finite-frequency type performance index function:
   \begin{equation}
   \label{FF_index}
   \sigma_{max}({G(\jmath\omega ) - {G_r}(\jmath\omega )}) ,   {\kern 15pt} for {\kern 4pt} all {\kern 4pt} \omega \in [\varpi_1, \varpi_2]
   \end{equation}
 Since the standard balanced truncation is intrinsically
 frequency-independent, hereby we will call it as \emph{frequency-independent balanced truncation} (FIBT) in the sequel, it cannot be used to further improve the in-band approximation performance with pre-known frequency information. To enhance the approximation performance over pre-specified frequency range, several balancing-related approaches have been developed. Some famous and popular ones include:

\indent (1) \emph{Singular perturbation approximation} (SPA). SPA is a companion balancing-related method of the standard FIBT and is first introduced by Liu and Anderson \cite{SPA_liu}. Although FIBT and SPA gives same entire-frequency type error bound, the characteristics of them are contrary to each other. The reduced systems generated by FIBT generally have a smaller error at high
frequencies, and tend to be larger at low frequencies. In contrast, SPA generally leads to good approximation performance at frequencies around $\omega=0$ by forcing the transfer function of full order model and reduced order model to be matched exactly at $\omega=0$ (i.e $G(\jmath 0)=G_r(\jmath 0)$). Therefore, SPA is particularly suited for solving model reduction problems in the cases that $\omega=0$ is pre-known as the dominating operating frequency point (\cite{SPA_Saragih} \cite{SPA_benner}). To further make the a flexible tradeoff between the local approximation performance over low-frequency ranges and the global approximation performance over entire frequency range, generalized SPA algorithm has been developed by introducing a user-defined adjustable scalar (see Obinata and Anderson \cite{SPA_Obinata}).

\indent (2) \emph{Frequency-weighted balanced truncation} (FWBT). In the fields of system analysis and control theory, frequency weighting functions is a conventional tool which has been widely applied for solving various analysis and synthesis problems with pre-known frequency information. For finite-frequency model order reduction problems,  utilizing the frequency weighting technique and combing it with the standard balanced truncation method also is very prevailing. During the last three decades, many frequency weighted balanced reduction approaches have been developed (see Enns \cite{FWBT_Enns1984}; Zhou \cite{FWBT_zhou1995}; Sreeram \cite{FWBT_sreeram_Agathoklis1989}; Ghafoor and Sreeram \cite{FWBT_Ghafoor_Sreeram2008}; Houlis and Sreeram \cite{FWBT_Houlis_Sreeram2009}; Wang et al \cite{FWBT_Wang_Sreeram_et_al1999}; Sreeram et al \cite{FWBT_Sreeram_et_al2005} and the references therein).  The common procedure of FWBT is build frequency-weighted model first by introducing input/out frequency weighted transfer functions and then apply the standard FIBT on the weighted model. Indeed, good frequency-specific approximation performance may be obtained if the selected weighting function is an \emph{appropriate} one. However, the design iterations to search for an \emph{appropriate} weighting transfer function can be tedious and time consuming. Besides, FWBT also suffers from the drawback of the increased order of the weighted
plant model.

\indent (3) \emph{Frequency-limited Grammians balanced truncation} (FGBT). It was first introduced by Gawronski and Juang in \cite{FGBT_Gawronski1990}. This methodology stems from the consideration of extending the definition of standard Gramians to the frequency-limited case and then applying the standard balanced truncation procedures to the frequency-limited Gramians (\cite{FGBT_Zadegan2005} \cite{FGBT_Shaker2014} \cite{FGBT_Ghafoor2008}). As has been pointed out in \cite{FGBT_Ghafoor2008},  FGBT may be invalid in some cases as the solutions of the ``frequency-limited Lyapunov equations" cannot be guaranteed to be positive semi-definite, and it provides no error bound. To overcome those drawbacks, several modified FGBT schemes providing error bound have been proposed (Gugercin and Antoulas \cite{BT2_Survey}; Gahfoor and Sreeram \cite{FGBT_Ghafoor2008})

A common feature of the those existing \emph{finite-frequency} balancing-related approaches is that they continue to use \emph{entire-frequency} type index (\ref{EF_index}) to evaluate the actually concerned \emph{finite-frequency} approximation performance (See Table I).
\begin{spacing}{1.0}
\begin{table}[!hbp]
\renewcommand{\arraystretch}{1.5}
\centering
\label{table1}
\begin{scriptsize}
 \caption{Characterizations of various balancing-related methods}\vspace{-0.1cm}
\begin{tabular}{|c|c|c|c|}
\Xhline{0.8pt}
  Assumption  & Method  & Actually concerned error   & Indices for the error bound\\ [2mm] \Xhline{0.8pt}
$\begin{array}{*{20}{c}}
  \rm{EF-MOR}  \\
  \omega \in [-\infty,+\infty] \\
\end{array}$ & FIBT  & $\sigma_{max}({G(\jmath \omega ) - {G_r}(\jmath \omega )}) ,   {\kern 2pt}\forall \omega \in [-\infty, +\infty]$
                      & $\sigma_{max}({G(\jmath \omega ) - {G_r}(\jmath \omega )}),   {\kern 2pt}\forall \omega \in [-\infty, +\infty]$  \\ [2mm]\Xhline{0.8pt}
\multirow{4}{*}{$\begin{array}{*{20}{c}}
  \\
  \\
  \rm{FF-MOR}  \\
  \omega \in [\varpi_1,+\varpi_2]  \\
\end{array}$} & SPA &  $\sigma_{max}({G(\jmath \omega ) - {G_r}(\jmath \omega )}) ,   {\kern 2pt}\forall \omega \in [\varpi_1, \varpi_2]$
                                 &  $\sigma_{max}({G(\jmath \omega ) - {G_r}(\jmath \omega )}),   {\kern 2pt}\forall \omega \in [-\infty, +\infty]$ \\ [2mm]
\cline{2-4}
& FWBT  & $\sigma_{max}({G(\jmath \omega ) - {G_r}(\jmath \omega )}),   {\kern 2pt}\forall \omega \in [\varpi_1, \varpi_2]$
            & $\sigma_{max}({G(\jmath \omega ) - {G_r}(\jmath \omega )}),   {\kern 2pt}\forall \omega \in[-\infty, +\infty]$\\[2mm]
\cline{2-4}
&  FGBT & $\sigma_{max}({G(\jmath \omega ) - {G_r}(\jmath \omega )}) ,   {\kern 2pt}\forall \omega \in [\varpi_1, \varpi_2]$
            & $\sigma_{max}({G(\jmath \omega ) - {G_r}(\jmath \omega )}),   {\kern 2pt}\forall \omega \in [-\infty, +\infty]$\\[2mm]
\cline{2-4}
&  $\begin{array}{*{20}{c}}
  \rm{FDBT} \\
  \rm{(To \; be\; developed)}\\
\end{array}$ & $\sigma_{max}({G(\jmath \omega ) - {G_r}(\jmath \omega )}) ,   {\kern 2pt}\forall \omega \in [\varpi_1, \varpi_2]$
            & $\sigma_{max}({G(\jmath \omega ) - {G_r}(\jmath \omega )}),   {\kern 2pt}\forall \omega \in [\varpi_1, \varpi_2]$ \\[2mm]
\Xhline{0.8pt}
\end{tabular}
\end{scriptsize}
\end{table}
\end{spacing}

As illustrated in Table I, there exists a incompatibleness between the intrinsic requirement and the achievement with respect to the existing finite frequency oriented balancing-related approaches. Since only entire-type error bounds are available, then whether or not the in-band approximation performance has been improved cannot be pre-known and guaranteed. In particular, FWBT and FGBT may gives rise to poor in-band approximation performance together with larger error bound in some cases. Moreover, there is little knowledge on the in-band approximation performance, even in the cases that the resulting in-band approximation performance is better than the standard FIBT method. This motivate us to revisit the finite-frequency model reduction problems.

In this paper, we are dedicated to deal with the finite-frequency model order reduction still within the framework of balanced truncation, however, a conceptual innovation that establishing \emph{finite-frequency} type error bound instead of \emph{entire-frequency }type error bound to estimate the in-band approximation error will be adopted in our development. The research scope and contribution of the present work is twofold. First, we focus on the cases that only a single dominating operating frequency point $\varpi$ is pre-known. By exploiting a special class of parameterized Mobious transformation, \emph{SF-type frequency-dependent balanced truncation} (FDBT) method was developed based on the Generalized KYP Lemma (Iwasaki and Hara \cite{Iwasaki2005AC}). It is shown that the proposed \emph{SF-type FDBT} provides a scalable SF-type error bound with respect to a user-defined parameter. By adjusting the parameter and picking it up with an appropriate value, it is probably to obtain satisfactory approximation performance. Second, we discuss the cases that both the upper bound and lower bound of operating frequency interval are pre-known. Following the same Generalized KYP Lemma based methodology, an \emph{interval-type frequency-dependent balanced truncation} method which provides interval-type error bound was developed. The \emph{interval-type FDBT} generally gives rise to good in-band approximation performance. In particular, we show that small in-band approximation error with small interval-type error bound could be simultaneously generated as long as the pre-specified interval is small enough.

The remainders of this paper is organized as follows: First, we introduce the Generalized KYP Lemma in Section 2. Then, we present the related results about \emph{SF-type frequency-dependent balanced truncation} method and \emph{interval-type frequency-dependent balanced truncation} method in Section 3 and Section 4, respectively. Next, we demonstrate the effectiveness and advantages of the proposed methods by several examples in Section 5.  Finally, we end with a conclusion in Section 6. \\

\noindent \textbf{Notations}: For a matrix $A$, $A^{T}$ and $A^{*}$ denote its transpose
and conjugate transpose, respectively. The symbol $*$
within a matrix represents the symmetric entries. $He(M)$ denotes $0.5(M+M^*)$. $\sigma_{max}(G)$
 denotes maximum singular value of the transfer matrix $G$. $Re(x)$ and $Im(x)$ denote the real part and imaginary part of the complex scalar $x$, respectively. $[M]^{\frac{1}{2}}$ denotes the square roots of matrix $M$ and $[M]^{\frac{1}{2}\star}$ denotes the \emph{ positive principle square root} of matrix $M$ (i.e. all the eigenvalues of $[M]^{\frac{1}{2}\star}$ has positive real part).  $I$ represents the identity matrix with appropriate dimension.\\

\section{Fundamental Tool}

The Kalman-Yakubovich-Popov (KYP) Lemma \cite{KYP} is a cornerstone in system and control theory. In fact, the EF-type error bound provided by the standard FIBT can be proofed and interpreted with the aid of KYP Lemma \cite{Zhoubook}. In \cite{Iwasaki2005AC}, Iwasaki and Hara successfully generalized the KYP Lemma from entire-frequency case to finite-frequency cases. The Generalized KYP Lemma plays a fundamental role in our developed and it is included here.\\

\begin{lemma} [Iwasaki and Hara \cite{Iwasaki2005AC}, Generalized KYP lemma] \label{lem-iwasaki} Consider a continuous-time system (\ref{originalsystem}), the following statements are equivalent:\\
(1) The frequency domain inequality
\begin{equation}
\label{lemma11} \sigma_{max}(G(\jmath \omega )) \le \gamma {\kern 4pt} holds {\kern 4pt} for {\kern 4pt} all{\kern 4pt} \omega {\kern 4pt} \in [\omega_1,\omega_2].
\end{equation}
 (2) There exist symmetric matrices $P$ and $Q$ of
appropriate dimensions, satisfying $Q> 0$ and 
\begin{small} { \begin{equation}
 \label{GKYP}
\pmatset{1}{0.36pt}
  \pmatset{0}{0.2pt}
  \pmatset{2}{10pt}
  \pmatset{3}{2pt}
  \pmatset{4}{2pt}
  \pmatset{5}{2pt}
  \pmatset{6}{2pt}
   \begin{pmat}[{|}]
  - He((j{\omega _1}I - A)Q{(j{\omega _2}I - A)^*})+ AP + P{A^*}  + B{B^*}  & (j{\omega _c}I - A)Q{C^*} + P{C^*} + B{D^*}
  \cr\-
   *                      & - CQ{C^*} + D{D^*} - {\gamma ^2}I \cr
  \end{pmat} \le 0.
\end{equation}}
\end{small}
\end{lemma}

\section{Frequency-dependent balanced truncation over uncertain frequency interval}

In this section, we focus on the model order reduction over an uncertain frequency interval (i.e. $ \omega \in [\varpi-\delta, \varpi+\delta ]$, where $\varpi$ denote the pre-known dominating frequency point, and $\delta$ denotes the unknown size of the frequency interval).  First, we construct a class of parameterized frequency-dependent extended systems, which plays an important role in the development of SF-type frequency-dependent balanced truncation.  Then, the related results and algorithm are presented. \\

\begin{definition}[SF-type Frequency-dependent Extend systems] \label{def-SF-FDES}  Given a system (\ref{originalsystem}) and a pre-specified frequency point $\varpi$,  the SF-type frequency-dependent extended systems can be constructed as:  \begin{equation}
\begin{small}
 \label{SF-FDES}
 G_{\epsilon\varpi}(\jmath \omega):\begin{array}{*{20}{c}}
  \pmatset{1}{0.1pt}
  \pmatset{0}{0.1pt}
  \pmatset{2}{4pt}
  \pmatset{3}{4pt}
  \pmatset{4}{4pt}
  \pmatset{5}{4pt}
  \pmatset{6}{1pt}
  {\begin{pmat}[{|}]
       A_\epsilon (\varpi) &  B_\epsilon (\varpi) \cr\-
       C_\epsilon (\varpi) &  D_\epsilon (\varpi) \cr
  \end{pmat}}=
  \pmatset{1}{0.1pt}
  \pmatset{0}{0.1pt}
  \pmatset{2}{4pt}
  \pmatset{3}{4pt}
  \pmatset{4}{4pt}
  \pmatset{5}{4pt}
  \pmatset{6}{1pt}
  {\begin{pmat}[{|}]
      \jmath \varpi I- \epsilon{(\epsilon I + \jmath \varpi I - A)^{ - 1}}(\jmath \varpi I - A) & \epsilon{(\epsilon I + \jmath \varpi I - A)^{ - 1}}B  \cr\-
      \epsilon C{(\epsilon I + \jmath \varpi I - A)^{ - 1}} & D+C{(\epsilon I + \jmath \varpi I - A)^{ - 1}}B  \cr
  \end{pmat}}
\end{array},\end{small}\end{equation}
\noindent \emph{where $\epsilon>0$ is a \textbf{user-specified scalar}. It
should be pointed out that $\epsilon$ should be a scalar satisfying
the condition: $\epsilon \neq -(\jmath \varpi -\lambda_i)$ to ensure
the invertibility of $(\epsilon I + \jmath \varpi I- A )$, where
$\lambda_i, i=1,...,n$ denote the eigenvalues of the matrix $A$.}\\
\end{definition}

\begin{proposition}\label{prop1} For a given system (\ref{originalsystem}), the corresponding SF-type frequency-dependent extended system (\ref{SF-FDES}) can be obtained by applying a particular Moebius transformation as follows:
\[\small {G_{\epsilon \varpi }}(\jmath \omega ) = G\left( {\frac{{a(\jmath \omega ) + b}}{{c(\jmath \omega ) + d}}} \right),\]
where $a = \epsilon - \jmath\varpi$, $b = -\varpi^2$, $c = -1$, $d = \epsilon + \jmath\varpi$.
\end{proposition}

\begin{proposition}\label{prop2} The following statements are true:\\
 a). If the original system (\ref{originalsystem}) is Hurwitz stable and $\epsilon>0$, then the corresponding SF-type frequency-dependent extended system is stable.  \\
 b). Given the original system (\ref{originalsystem}) is unstable and denote the unstable eigenvalues of $A$ as $\lambda_i^{+}, i=1,...,n_u$, then the corresponding SF-type frequency-dependent extended system is stable if the value of $\epsilon$ satisfying $0<\epsilon<min({\epsilon_i^+}), i=1,...,n_u$, where
$\epsilon_i^+=(\varpi-Im(\lambda_i))^2/Re(\lambda_i)+Re(\lambda_i)$.
\end{proposition}
\begin{proof}  a). Let us denote $\lambda_i,i=1,2...,n$, and
$\lambda_{\epsilon i}(\varpi),i=1,2...,n$ as the eigenvalues of the matrices
$A$ and $A_{\epsilon}(\varpi)$, respectively. According to the mapping between
$A$ and $A_{\epsilon}(\varpi)$ given in (\ref{SF-FDES}), we know that
\[\lambda_{\epsilon i} (\varpi)=\jmath \varpi-\epsilon(\jmath \varpi -\lambda_i)(\epsilon+\jmath
\varpi-\lambda_i)^{-1}, i=1,...,n\]
\noindent Noticing that $Re(\lambda_i)<0$ if the system $G(\jmath \omega)$ is stable, then the following inequalities
\begin{equation}
\label{ProofPro2}
Re(\lambda_{\epsilon i}(\varpi)=-\frac{-\epsilon Re(\lambda_i)(\epsilon-Re(\lambda_i))+\epsilon(\varpi-Im(\lambda_i))^2}{(\epsilon-Re(\lambda_i))^2+(\varpi-Im(\lambda_i))^2}<0, i=1,...n\end{equation}
hold if $\epsilon>0$. Thus the proof is completed.\\[4mm]
\noindent b). Denote $\lambda_{\epsilon i}^+(\varpi), i=1,...,n_u$ as the eigenvalues of $A_{\epsilon}(\varpi)$ mapped from $\lambda_{i}^+$, i.e.
  \[\lambda_{\epsilon i}^+ (\varpi)=\jmath \varpi-\epsilon(\jmath \varpi -\lambda_i^+)(\epsilon+\jmath
\varpi-\lambda_i^+)^{-1}, i=1,...,n_u\]
then it can be concluded that $Re(\lambda_{\epsilon i}^+(\varpi))<0, i=1,...,n_u$ for all $\epsilon$ satisfying $0<\epsilon<min({\epsilon_i^+}), i=1,...,n_u$, according to the computational formula (\ref{ProofPro2}). Thus the proof is completed. \end{proof}

\begin{definition}[SF-type Frequency-dependent Lyapunov Equations]\label{def-SF-FDLE} Given a linear continuous-time system (\ref{originalsystem})
and one of its corresponding Hurwitz stable SF-type frequency-dependent
extended systems (\ref{SF-FDES}), then the following
two Lyapunov equation
\begin{equation}
\label{SF-FDLE}
\begin{array}{l}
 A_\epsilon (\varpi) W_{c\epsilon} (\varpi)  + W_{c\epsilon} (\varpi)A^*_\epsilon(\varpi)+  B_\epsilon(\varpi)B^*_\epsilon(\varpi)=0,\\
 A_\epsilon^* (\varpi)W_{o\epsilon} (\varpi) + W_{o\epsilon} (\varpi)A_\epsilon(\varpi)  +  C^*_\epsilon(\varpi)C_\epsilon(\varpi)=0. \\
 \end{array}\end{equation}
are defined as \emph{SF-type frequency-dependent controllability and
observability Lyapunov equations} of the continuous-time system
(\ref{originalsystem}). Furthermore, the solutions $W_{c\epsilon}(\varpi)$ and $
W_{o\epsilon}(\varpi)$ will be referred to as \emph{SF-type frequency-dependent controllability and observability Gramians} of the continuous-time system (\ref{originalsystem}).
\end{definition}

\begin{definition}[SF-type Frequency-dependent Balanced Realization]\label{SF-FDBR} Given a linear continuous-time system (\ref{originalsystem})
and one of its Hurwitz stable SF-type frequency-dependent extended systems
(\ref{SF-FDES}),  the corresponding SF-type frequency-dependent
controllability and observability Gramians are equal and diagonal,
i.e. the following Lyapunov equations
\begin{equation}
\label{SF-FDLE-balanced}
\begin{array}{l}
 A_\epsilon (\varpi) \Sigma_\epsilon(\varpi)+ \Sigma_{\epsilon} (\varpi) A^*_\epsilon(\varpi)+B_{\epsilon}(\varpi)B^*_{\epsilon}(\varpi)=0,\\
 A^*_\epsilon (\varpi)\Sigma_\epsilon(\varpi)+\Sigma_{\epsilon}(\varpi)  A_\epsilon(\varpi)  +C^*_{\epsilon}(\varpi)C_{\epsilon}(\varpi)=0. \\
 \end{array}\end{equation}
simultaneously hold, then this particular realization will be
referred to as a \emph{SF-type frequency-dependent balanced realization} \end{definition}

\begin{proposition} Suppose the given system (\ref{originalsystem}) is stable and let $W_c, W_o,\Sigma$ denote its standard controllability and observability and balanced Gramian matrices, then the following statements are true:\\
{\small \emph{a).} $W_c>W_{c\epsilon}(\varpi)$,
$W_o>W_{o\epsilon}(\varpi)$, $\Sigma>\Sigma_\epsilon(\varpi)$,\\
\emph{b).} $\mathop {\lim }\limits_{\varepsilon  \to 0} W_{c\epsilon}(\varpi)
 = 0, {\kern 2pt} \mathop {\lim }\limits_{\varepsilon  \to 0}
W_{o\epsilon}(\varpi) = 0, {\kern 2pt} \mathop {\lim }\limits_{\varepsilon
\to 0} {\Sigma_\epsilon(\varpi)
} = 0 $,\\
\emph{c).} $\mathop {\lim }\limits_{\varepsilon  \to \infty}
{W_{c\epsilon}(\varpi) } = W_c, {\kern 2pt} \mathop {\lim }\limits_{\varepsilon
\to \infty} {W_{o\epsilon}(\varpi) } = W_o, {\kern 2pt} \mathop {\lim
}\limits_{\varepsilon \to \infty} {\Sigma_\epsilon(\varpi) } = \Sigma$.}
\end{proposition}
\begin{proof} a). It is well known that the standard controllability and
observability Gramian matrices $W_c, W_o$ of system (1) satisfy the
following standard frequency-independent Lyapunov equations:
\begin{equation}
\small
\label{FILE}\begin{array}{l}
AW_{c}+W_{c}A^*+BB^*=0\\
A^*W_{o}+W_{o}A+C^*C=0. \\
 \end{array}\end{equation}
Post-and-pre multiply the SF-type frequency-dependent Lyapunov equations
(\ref{SF-FDLE}) by $\epsilon^{-1}(\epsilon I+ \jmath
\varpi I-A)$, then we have
\begin{equation}
\small
\label{ww-lyapunov_equation}\begin{array}{l}
 AW_{\varpi c} +  {W_{\varpi c}}{A^*} + 2{\epsilon}^{-1} (\jmath \varpi I - A){W_{\varpi c}}{(\jmath \varpi I - A)^*} +BB^*=0 \\
 {A^*} {W_{\varpi o}} + {W_{\varpi o}}A + 2{\epsilon}^{-1}{(\jmath \varpi I - A)^*}{W_{\varpi o}}(\jmath \varpi I - A) +BB^*= 0. \\
 \end{array}\end{equation}
Furthermore, the following equations can be derived by subtracting the
equations (\ref{FILE}) from
(\ref{ww-lyapunov_equation})
\begin{equation}
\small
\label{www-lyapunov_equation}\begin{array}{l}
 A({W_c} - W_{c\epsilon}(\varpi)) + ({W_c} - W_{c\epsilon}(\varpi)){A^*} + 2{\epsilon}^{-1} (\jmath \varpi I - A)W_{c\epsilon}(\varpi){(\jmath \varpi I - A)^*} = 0 \\
 {A^*}({W_o} - W_{o\epsilon}(\varpi)) + ({W_o} - W_{o\epsilon}(\varpi))A + 2{\epsilon}^{-1}{(\jmath \varpi I - A)^*}W_{o\epsilon}(\varpi)(\jmath \varpi I - A) = 0 \\
 \end{array}\end{equation}
It is easily to conclude that $({W_c} - W_{o\epsilon}(\varpi))>0$ and
$({W_o} - W_{o\epsilon}(\varpi))>0$ since
\begin{equation}
\small
\begin{array}{l}
 2{\epsilon}^{-1} (\jmath \varpi I - A)W_{c\epsilon}(\varpi){(\jmath \varpi I - A)^*}>0 \\
2{\epsilon}^{-1}{(\jmath \varpi I - A)^*}W_{o\epsilon}(\varpi)(\jmath \varpi I - A)>0. \\
 \end{array}\end{equation}
 \noindent Thus the proof is completed.\\
 \noindent b). The SF-type frequency-dependent Lyapunov equations (\ref{SF-FDLE}) can be rewritten as:
\begin{equation}
\small
\begin{array}{l}
He((\jmath \varpi I - A){(\epsilon I + \jmath \varpi I - A)^{ - 1}}{ W_{c\epsilon}(\varpi)}) = \epsilon{(\epsilon I + \jmath \varpi I - A)^{ - 1}}B{B^*}{(\epsilon I + \jmath \varpi I - A)^{ - *}} \\
He((\jmath \varpi I - A)^*(\epsilon I + \jmath \varpi I - A)^{ - *} W_{o\epsilon}(\varpi))   = \epsilon(\epsilon I + \jmath \varpi I - A)^{ - *}{C^*}C(\epsilon I + \jmath \varpi I - A)^{ - 1}. \\
 \end{array}\end{equation}
thus one can conclude that:
\[
\small
\begin{array}{l}
 \mathop {\lim }\limits_{\epsilon  \to 0}  W_{c\epsilon(\varpi) } = \frac{1}{2}  He( \mathop {\lim }\limits_{\epsilon  \to 0} (\jmath \varpi I - A){(\epsilon I + \jmath \varpi I - A)^{ - 1}}{ \mathop {\lim }\limits_{\epsilon  \to 0} W_{c\epsilon}(\varpi)}) \\
{\kern 52pt} =   \frac{1}{2} \mathop {\lim }\limits_{\epsilon  \to 0} { \epsilon}\mathop {\lim }\limits_{\epsilon  \to 0} (\epsilon I +\jmath \varpi I - A{)^{ - 1}}B{B^*}(\epsilon I+\jmath \varpi I - A)^{ - *}=0, \\
 \mathop {\lim }\limits_{\epsilon  \to 0}  W_{o\epsilon(\varpi) } =  \frac{1}{2}  He( \mathop {\lim }\limits_{\epsilon  \to 0} (\jmath \varpi I - A)^*{(\epsilon I + \jmath \varpi I - A)^{ - *}}{ \mathop {\lim }\limits_{\epsilon  \to 0} W_{o\epsilon}(\varpi)}) \\
 {\kern 52pt} = \frac{1}{2}\mathop {\lim }\limits_{\epsilon  \to 0}  {\epsilon}\mathop {\lim }\limits_{\epsilon  \to 0}  (\epsilon I+\jmath \varpi I - A{)^{ - *}}{C^*}{C}(\epsilon I+\jmath \varpi I - A)^{ - {1}}=0.\\
 \end{array}\]
 \noindent Thus the proof is completed.\\
\noindent 3). It can be easily observed that the
$\varpi$-dependent matrices ${A_\varpi }, {B_\varpi }, {C_\varpi }$
will recover $A,B,C$ as $\epsilon \to \infty$, i.e.
\begin{equation}
\small
\begin{array}{l} \mathop {\lim }\limits_{\varepsilon  \to
\infty } {A_\varpi } =\mathop {\lim }\limits_{\varepsilon  \to
\infty } {(\jmath \varpi I- \epsilon{(\epsilon I + \jmath \varpi I -
A)^{ - 1}}(\jmath \varpi I - A))}=
A,\\
\mathop {\lim }\limits_{\varepsilon  \to \infty } {B_\varpi }
=\mathop {\lim }\limits_{\varepsilon  \to \infty } {
\epsilon{(\epsilon I + \jmath \varpi I - A)^{ - 1}}B  }=
B,\\
\mathop {\lim }\limits_{\varepsilon  \to \infty } {C_\varpi }
=\mathop {\lim }\limits_{\varepsilon  \to \infty } { \epsilon
C{(\epsilon I + \jmath \varpi I - A)^{ - 1}} } =
C.\\
 \end{array}\end{equation}
\noindent Then it is trivial to conclude that
 \[\begin{array}{l}
\mathop {\lim }\limits_{\epsilon \to \infty} {W_\varpi }_c = W_c,
{\kern 4pt} \mathop {\lim }\limits_{\epsilon \to \infty} {W_\varpi
}_o = W_o, {\kern 4pt} \mathop {\lim }\limits_{\epsilon \to \infty}
{\Sigma_\varpi} =\Sigma. \\
 \end{array}\]
 \end{proof}

\begin{theorem}[SF-type Frequency-dependent Balanced Truncation]\label{the-SF-FDBT} Given a linear continuous-time system (\ref{originalsystem})
and the pre-known dominating operating frequency point $\omega=\varpi$, then for any one
of its Hurwitz stable SF-type frequency-dependent extended systems
(\ref{SF-FDES}) given in SF-type frequency-dependent balanced
realization with respect to the SF-type frequency-dependent Gramian
$\Sigma_{\epsilon}(\varpi)=diag{(\Sigma_{\epsilon 1}(\varpi),\Sigma_{\epsilon 2}(\varpi))}$
\[\begin{array}{l}
\Sigma_{\epsilon 1}(\varpi)=diag{(\sigma_{\epsilon 1}(\varpi),    \sigma_{\epsilon 2}(\varpi),...,\sigma_{\epsilon r}(\varpi))}, \\
\Sigma_{\epsilon 2}(\varpi)=diag{(\sigma_{\epsilon (r+1)}(\varpi),\sigma_{\epsilon (r+2)} (\varpi),...,\sigma_{\epsilon n}(\varpi))}, \\
\end{array}\]
 and $\sigma_{\epsilon 1}(\varpi)\geq
...\geq\sigma_{\epsilon r}(\varpi)\geq...\geq\sigma_{\epsilon n}(\varpi),$
 the desired $r^{th}$-order model  $
  \pmatset{1}{0.36pt}
  \pmatset{0}{0.2pt}
  \pmatset{2}{2pt}
  \pmatset{3}{2pt}
  \pmatset{4}{2pt}
  \pmatset{5}{2pt}
  \pmatset{6}{2pt}
G_{r}(\jmath \omega):= {   \begin{pmat}[{|}]
       A_{r}   &  B_{r} \cr\-
       C_{r}   &  D_{r} \cr
  \end{pmat}}
$  is given by: 
\begin{equation}
\label{reducedmodel-SF}
\begin{array}{l}
  A_{r}=   \jmath \varpi I-\epsilon Z_r(\jmath \varpi I-A_{\epsilon}(\varpi)){Z_r^T} (\epsilon I -Z_r(\jmath \varpi I-A_{\epsilon}(\varpi)){Z_r^T})^{-1},\\
  B_{r} =  \epsilon^{-1}(\epsilon  I + \jmath \varpi I - A_{r})Z_r B_{\epsilon} (\varpi), \\
  C_{r} =  \epsilon^{-1} C_{\epsilon} (\varpi)Z_r^T(\epsilon I + \jmath \varpi I - A_{r}),\\
  D_{r} =  D_{\epsilon}(\varpi)-C_{r}(\epsilon  I + \jmath \varpi I - A_{r})^{-1}B_{r}, \\
 \end{array}\end{equation}

\noindent  where $
 {Z_r} =[I^{r \times r} {\kern 6pt}  0^{r \times (n - r)}]
$. Furthermore, the truncated model $G_r(\jmath \omega)$ possesses
the following properties: \\
\noindent 1). The approximation error between the original system model
(\ref{originalsystem}) and the truncated $r^{th}$ reduced model (\ref{reducedmodel-SF})
at the given frequency point $\omega=\varpi$ satisfies the following SF-type error bound: 
\begin{equation}
  \label{SF-EB-SF}\sigma_{max}(G(\jmath \omega ) - {G_{r}}(\jmath \omega)) \le 2\sum\limits_{i = r + 1}^{n} {{\sigma _{i\epsilon}(\varpi)}}, {\kern 8pt}  for {\kern 6pt} \omega=\varpi.
\end{equation}
 2). The approximation error between the original system model
(\ref{originalsystem}) and the truncated $r^{th}$ reduced model (\ref{reducedmodel-SF})
over entire frequency range satisfies the following EF-type error bound: 
\begin{equation}
\label{EF-EB-SF}
\begin{array}{l}
\sigma_{max}(G(\jmath \omega ) - {G_r}(\jmath \omega))
  \le 2 \sum\limits_{i = r + 1}^{n} {{\sigma _{i\varpi}}} {\kern 1pt}\\
 {\kern 114pt}+{\left\| G(\jmath \omega)-G_{\epsilon \varpi}(\jmath \omega) \right\|_\infty }\\
 {\kern 114pt}+{\left\| G_r(\jmath \omega)-G_{r \epsilon \varpi}(\jmath \omega) \right\|_\infty }, \; for \; all \; \omega \in [-\infty,+\infty] \\
 \end{array}\end{equation}
\noindent where
\begin{equation}
\begin{small}
 \label{Reduced_SF-FDES}
 G_{r \epsilon\varpi}(\jmath \omega):\begin{array}{*{20}{c}}
  \pmatset{1}{0.1pt}
  \pmatset{0}{0.1pt}
  \pmatset{2}{4pt}
  \pmatset{3}{4pt}
  \pmatset{4}{4pt}
  \pmatset{5}{4pt}
  \pmatset{6}{1pt}
  {\begin{pmat}[{|}]
       A_{r\epsilon} (\varpi) &  B_{r\epsilon} (\varpi) \cr\-
       C_{r\epsilon} (\varpi) &  D_{r\epsilon} (\varpi) \cr
  \end{pmat}}=
  \pmatset{1}{0.1pt}
  \pmatset{0}{0.1pt}
  \pmatset{2}{4pt}
  \pmatset{3}{4pt}
  \pmatset{4}{4pt}
  \pmatset{5}{4pt}
  \pmatset{6}{1pt}
  {\begin{pmat}[{|}]
      \jmath \varpi I- \epsilon{(\epsilon I + \jmath \varpi I - A_r)^{ - 1}}(\jmath \varpi I - A_r) & \epsilon{(\epsilon I + \jmath \varpi I - A_r)^{ - 1}}B_r  \cr\-
      \epsilon C_r{(\epsilon I + \jmath \varpi I - A_r)^{ - 1}} & D_r+C_r{(\epsilon I + \jmath \varpi I - A_r)^{ - 1}}B_r  \cr
  \end{pmat}}
\end{array},\end{small}\end{equation}
\end{theorem}

\begin{proof} 1). The detailed proof for $r=n-1$ case will be
provided in the sequel, and the $r=n-2,...1$ cases can be
easily completed step by step. \\[-8mm]

\noindent The error system model between the original high-order system model $G(\jmath \omega)$ and
the truncated $(n-1)^{th}$ reduced model $G_{n-1}(\jmath \omega)$ can be
represented by
\begin{equation}\begin{array}{l}
\label{errorsystem}
 \mathcal E_n(\jmath \omega)=G(\jmath \omega)-G_{n-1}(\jmath \omega)=: \pmatset{1}{0.36pt}
  \pmatset{0}{0.2pt}
  \pmatset{2}{8pt}
  \pmatset{3}{8pt}
  \pmatset{4}{8pt}
  \pmatset{5}{4pt}
  \pmatset{6}{4pt}    \begin{pmat}[{|}]
        \mathcal A_{en}   &    \mathcal B_{en} \cr\-
        \mathcal C_{en}   &    \mathcal D_{en} \cr
      \end{pmat}    =
  \pmatset{1}{0.36pt}
  \pmatset{0}{0.2pt}
  \pmatset{2}{4pt}
  \pmatset{3}{4pt}
  \pmatset{4}{4pt}
  \pmatset{5}{2pt}
  \pmatset{6}{2pt}
 \begin{pmat}[{.|}]
      A_{n-1}   & 0  &   B_{n-1}\cr
      0     & A  &   B  \cr\-
       -C_{n-1} & C  &   D-D_{n-1}  \cr
      \end{pmat}     \\
 \end{array}.\end{equation}\\[-8mm]
From the error system $  \mathcal E_n(\jmath \omega)$, we can
construct a dilated system $\mathscr E_n(\jmath \omega)$ as follow:\begin{equation}
  \begin{array}{l}
  \pmatset{1}{0.36pt}
  \pmatset{0}{0.2pt}
  \pmatset{2}{8pt}
  \pmatset{3}{8pt}
  \pmatset{4}{8pt}
  \pmatset{5}{4pt}
  \pmatset{6}{4pt}
\mathscr E_n(\jmath \omega)=\begin{pmat}[{|}]
      \mathscr A_{en}  & \mathscr B_{en}   \cr\-
      \mathscr C_{en}  & \mathscr D_{en}   \cr
      \end{pmat} =  \pmatset{1}{0.36pt}
  \pmatset{0}{0.2pt}
  \pmatset{2}{6pt}
  \pmatset{3}{6pt}
  \pmatset{4}{6pt}
  \pmatset{5}{4pt}
  \pmatset{6}{4pt}\begin{pmat}[{|.}]
       \mathcal A_{en}  &  \mathcal B_{en}      &  \mathcal B_{dn} \cr\-
       \mathcal C_{en}  &  \mathcal D_{en}      &  \mathcal D_{dn}^{11} \cr
       \mathcal C_{dn}  &  \mathcal D_{dn}^{12}       &   \mathcal D_{dn}^{22} \cr
      \end{pmat} \\
  \end{array},
\end{equation}
\noindent where $\mathcal B_{dn}, \mathcal C_{dn}, \mathcal D_{dn}^{12},
\mathcal D_{dn}^{21}, \mathcal D_{dn}^{22}$ are auxiliary 'dilated' matrices,
and those matrices are constructed as follows:\begin{equation}
\begin{array}{l}
 {  \mathcal B_{dn}} =  -   \sigma_{\epsilon n}(\varpi) \epsilon^{-1}(\epsilon I+\jmath \varpi I- \mathcal A_{en})  \left[ \begin{array}{l}
 Z_{n-1} \\
 -I \\
 \end{array} \right]
 {\Sigma ^{ - 1}_{\epsilon}(\varpi)}{C_{\epsilon}(\varpi)^*}, \\[4mm]
 {  \mathcal C_{dn}}^* =  - \sigma_{\epsilon n}(\varpi)\epsilon^{-1}(\epsilon I+\jmath \varpi I- \mathcal A_{en})^{T}\left[ \begin{array}{l}
 -Z_{n-1} \\
 -I \\
 \end{array} \right]
 {\Sigma ^{ - 1}_{\epsilon}(\varpi)}B_{\epsilon}(\varpi), \\[4mm]
  \mathcal D_{dn}^{12}  =  - {  \mathcal C}_{en} (\epsilon I+\jmath \varpi I- \mathcal A_{en})^{-1}  \mathcal B_{dn}  +2\sigma_{\epsilon n}(\varpi)  I,\\[2mm]
  \mathcal D_{dn}^{21}  =  - {  \mathcal C}_{dn} (\epsilon I+\jmath \varpi I- \mathcal A_{en})^{-1}  \mathcal B_{en}  +2\sigma_{\epsilon n}(\varpi)   I,\\[2mm]
  \mathcal D_{dn}^{22} =  -  {  \mathcal C}_{dn} (\epsilon I+\jmath \varpi I- \mathcal A_{en})^{-1}  \mathcal B_{dn}.\\
 \end{array}\end{equation}
\noindent Defining the Lyapunov variable $\mathscr Q_{en}=\mathscr Q_{en}^*  \ge 0$ and
$\mathscr P_{en}=\mathscr P_{en}$ as follows:
\begin{equation}
\small
\begin{array}{l}
 \mathscr Q_{en}
 =2 {\epsilon }^{-1} \left[ \begin{array}{l}
 Z_{n-1} \\
 I \\
 \end{array} \right]\Sigma_{\epsilon}(\varpi)  \left[ \begin{array}{l}
 Z_{n-1} \\
 I \\
 \end{array} \right]^T+2{\epsilon }^{-1}\sigma_{\epsilon n}(\varpi)^2   \left[ \begin{array}{l}
 -Z_{n-1} \\
 I \\
 \end{array} \right] \Sigma_\epsilon^{-1}(\varpi) \left[ \begin{array}{l}
 -Z_{n-1} \\
 I \\
 \end{array} \right]^T ,\\[4mm]
 \mathscr P_{en} = \left[ \begin{array}{l}
 Z_{n-1} \\
 I \\
 \end{array} \right] \Sigma_{\epsilon}(\varpi)  \left[ \begin{array}{l}
 Z_{n-1} \\
 I \\
 \end{array} \right]^T  +
\sigma_{\epsilon n}(\varpi)^2\left[ \begin{array}{l}
 -Z_{n-1} \\
 I \\
 \end{array} \right] \Sigma_{\epsilon}^{-1}(\varpi)  \left[ \begin{array}{l}
 -Z_{n-1} \\
 I \\
 \end{array} \right]^T.{\kern 40pt}
 \end{array}\end{equation}\\
\noindent  Substituting the above constructed Lyapunov variable
$\mathscr Q_{en}, \mathscr P_{en}$ into the following SF-type matrix inequality suggested by the Generalized KYP Lemma,
 \begin{equation}
\small
\begin{array}{l}
 \pmatset{1}{0.1pt}
  \pmatset{0}{0.1pt}
  \pmatset{2}{4pt}
  \pmatset{3}{4pt}
  \pmatset{4}{4pt}
  \pmatset{5}{4pt}
  \pmatset{6}{1pt}
  {\kern 10pt}{\begin{pmat}[{|}]
    -(\jmath \varpi I-\mathscr A_{en}) \mathscr Q_{en} (\jmath \varpi-\mathscr A_{en})^*+ \mathscr A_{en} \mathscr P_{en} + \mathscr P_{en}\mathscr A_{en}^*  + \mathscr B_{en} {\mathscr B_{en}^*}  & (\jmath \varpi I -\mathscr A_{en}) \mathscr C_{en}^*+ \mathscr P_{en}{\mathscr C_{en}^*} + \mathscr B_{en}{\mathscr D_{en}^*} \cr\-
      *  & -\mathscr C_{en}\mathscr Q_{en}\mathscr C_{en}^*+ \mathscr D_{en}{\mathscr D_{en}^*} - (2\sigma_{\epsilon n}(\varpi))^2I  \cr
  \end{pmat}} \\[4mm]
= \pmatset{1}{0.36pt}
  \pmatset{0}{0.2pt}
  \pmatset{2}{4pt}
  \pmatset{3}{4pt}
  \pmatset{4}{4pt}
  \pmatset{5}{4pt}
  \pmatset{6}{4pt}
  {\begin{pmat}[{|}]
      \Pi_{11}      & \Pi_{12}   \cr\-
      *           & \Pi_{22}    \cr
  \end{pmat}} \\
=  \pmatset{1}{0.36pt}
  \pmatset{0}{0.2pt}
  \pmatset{2}{4pt}
  \pmatset{3}{4pt}
  \pmatset{4}{4pt}
  \pmatset{5}{4pt}
  \pmatset{6}{4pt}
  {\begin{pmat}[{|.}]
      \Pi_{11}      & \Pi_{12}^{1}   & \Pi_{12}^{2}\cr\-
      *           & \Pi_{22}^{11}    & \Pi_{22}^{12}\cr
      *           & *                & \Pi_{22}^{22}\cr
  \end{pmat}}\\
 \end{array}\end{equation}

\noindent Combing the balanced SF-type frequency-dependent Lyapunov equations (\ref{SF-FDLE-balanced}), one can derive the following equations:
\begin{equation}
\label{ProofTheorem1_equ1}
\begin{array}{l}
 {\Pi _{11}}  = -(\jmath \varpi I-\mathscr A_{en}) \mathscr Q_{en} (\jmath \varpi-\mathscr A_{en})^*+ \mathscr A_{en} \mathscr P_{en} + \mathscr P_{en}\mathscr A_{en}^*  + \mathscr B_{en} {\mathscr B_{en}^*} \\ [2mm]
 {\kern 18pt}=[\epsilon^{-1} (\epsilon I+ \jmath \varpi I -\mathcal A_{en})]   \Delta_1  [\epsilon^{-1} (\epsilon I+ \jmath \varpi I -\mathcal A_{en})]^*\\
 \end{array}{\kern 58pt}\end{equation}
\begin{equation}
\label{ProofTheorem1_equ2}
\begin{array}{l}
 {\Pi _{12}^1}  = (\jmath \varpi I -\mathcal A_{en}) \mathscr Q_{en} \mathcal C_{en}^*+ \mathscr P_{en}{\mathcal C_{en}^*} + \mathcal B_{en}{\mathcal D_{en}^*}      \\[2mm]
 {\kern 18pt} = [\epsilon (\epsilon I+ \jmath \varpi I -\mathcal A_{en})]  \Delta_2   [\epsilon (\epsilon I+\jmath \varpi I -A)^{-1}]^*  C^*   \\
 \end{array}{\kern 146pt}
 \end{equation}
\begin{equation}
\label{ProofTheorem1_equ3}
\begin{array}{l}
 {\Pi _{12}^2}   = (\jmath \varpi I -\mathcal A_{en}) \mathscr Q_{en} \mathcal C_{en}^*+ \mathscr P_{en}{\mathcal C_{en}^*} + \mathcal B_{en}{\mathcal D_{en}^*}\\
  {\kern 18pt} = [\epsilon^{-1} (\epsilon I+ \jmath \varpi I -\mathcal A_{en})]   \Delta_3  \Sigma_{\epsilon}^{-1}(\varpi) [\epsilon (\epsilon I+\jmath \varpi I -A)^{-1}]  B \\
 \end{array}
 {\kern 106pt}
 \end{equation}
 \begin{equation}
\label{ProofTheorem1_equ4}\begin{array}{l}
 \pmatset{1}{0.36pt}
  \pmatset{0}{0.2pt}
  \pmatset{2}{4pt}
  \pmatset{3}{4pt}
  \pmatset{4}{18pt}
  \pmatset{5}{3pt}
  \pmatset{6}{3pt}
 {\Pi _{22}^{11}}= -\mathcal C_{en}\mathscr Q_{en} \mathcal C_{en}^*+  {\begin{pmat}[{|}]
\mathcal D_{en} & \mathcal D_{dn}^{12}  \cr
  \end{pmat}}{\begin{pmat}[{|}]
\mathcal D_{en} & \mathcal D_{dn}^{12}  \cr
  \end{pmat}^*} - (2 \sigma_{\epsilon n} (\varpi))^2 I \\[3mm]
 {\kern 18pt}=\pmatset{1}{0.1pt}
  \pmatset{0}{0.1pt}
  \pmatset{2}{1pt}
  \pmatset{3}{1pt}
  \pmatset{4}{1pt}
  \pmatset{5}{1pt}
  \pmatset{6}{1pt}
   -{\epsilon ^{ - 1}}{\mathcal C_{en}}\Delta_2 { [\epsilon (\epsilon I+\jmath \varpi I -A)^{-1}]^*}{C^*}\\ [2mm]
  {\kern 32pt}  - {\epsilon ^{ - 1}}C[\epsilon (\epsilon I+\jmath \varpi I -A)^{-1}] \Delta_2^* {\mathcal C_{en}}^*  \\
 \end{array}{\kern 118pt}\end{equation}
 \begin{equation}
\label{ProofTheorem1_equ5}
 \pmatset{1}{0.36pt}
  \pmatset{0}{0.2pt}
  \pmatset{2}{4pt}
  \pmatset{3}{4pt}
  \pmatset{4}{18pt}
  \pmatset{5}{3pt}
  \pmatset{6}{3pt}
\begin{array}{l}
 {\Pi _{22}^{12}} = -\mathcal C_{en}\mathscr Q_{en}\mathcal C_{en}^*+  {\begin{pmat}[{|}]
  \mathcal D_{en}      & \mathcal D_{dn}^{12}  \cr
  \end{pmat}}{\begin{pmat}[{|}]
  \mathcal D_{dn}^{21} & \mathcal D_{dn}^{22}  \cr
  \end{pmat}^*}  \\[3mm]
 {\kern 18pt} = - {\epsilon ^{ - 1}}C[\epsilon (\epsilon I+\jmath \varpi I -A)^{-1}] \Delta_2 \mathcal C_{dn}^* \\ [2mm]
 {\kern 30pt}  - {\epsilon ^{ - 1}}\sigma_{\epsilon n}(\varpi) {\mathcal C_{en}}\Delta_3 {\Sigma_{\epsilon}^{ - 1}(\varpi)}[\epsilon (\epsilon I+\jmath \varpi I -A)^{-1}] B      \\
 \end{array}{\kern 142pt}\end{equation}
 \begin{equation}
\label{ProofTheorem1_equ6}
\begin{array}{l}
 \pmatset{1}{0.36pt}
  \pmatset{0}{0.2pt}
  \pmatset{2}{4pt}
  \pmatset{3}{4pt}
  \pmatset{4}{18pt}
  \pmatset{5}{3pt}
  \pmatset{6}{3pt}
 {\Pi _{22}^{22}}  = -\mathcal C_{en}\mathscr Q_{en}\mathcal C_{en}^*+   \pmatset{0}{0.1pt}
 {\begin{pmat}[{|}]
\mathcal D_{dn}^{21} & \mathcal D_{dn}^{22}  \cr
  \end{pmat}}{\begin{pmat}[{|}]
\mathcal D_{dn}^{21} & \mathcal D_{dn}^{22}  \cr
  \end{pmat}^*} - (2 \sigma_{\epsilon n} (\varpi))^2 I \\[3mm]
 {\kern 18pt}=\pmatset{1}{0.1pt}
  \pmatset{0}{0.1pt}
  \pmatset{2}{1pt}
  \pmatset{3}{1pt}
  \pmatset{4}{1pt}
  \pmatset{5}{1pt}
  \pmatset{6}{1pt}
 -{\epsilon ^{ - 1}}{\sigma_{\epsilon n}^2(\varpi)}B^*[\epsilon (\epsilon I+\jmath \varpi I -A)^{-1}]^*{\Sigma_{\epsilon}^{ - 1}(\varpi)} \Delta_3 {\Sigma_{\epsilon}^{ - 1}(\varpi)}[\epsilon (\epsilon I+\jmath \varpi I -A)^{-1}] B\\[2mm]
  {\kern 30pt} - {\epsilon ^{ - 1}}{\sigma_{\epsilon n}^2(\varpi)}  B^*[\epsilon (\epsilon I+\jmath \varpi I -A)^{-1}]^*{\Sigma_{\epsilon}^{ - 1}(\varpi)} \Delta_3^* {\Sigma_{\epsilon}^{ - 1}(\varpi)} [\epsilon (\epsilon I+\jmath \varpi I -A)^{-1}] B     \\
 \end{array} \end{equation}

where

\begin{equation}
\small
\label{ProofTheorem1_equ7}
\begin{array}{l}
 {\Delta _1} \\
  =\pmatset{1}{0.1pt}
  \pmatset{0}{0.1pt}
  \pmatset{2}{1pt}
  \pmatset{3}{1pt}
  \pmatset{4}{1pt}
  \pmatset{5}{1pt}
  \pmatset{6}{1pt}
    {\begin{pmat}[{|}]
   A_{r\epsilon}(\varpi) &  0\cr\-
   0   & A_{\epsilon}(\varpi)  \cr
  \end{pmat}}
    \left({\begin{pmat}[{|}]
   Z_{n-1}\Sigma_\epsilon (\varpi) {Z_{n-1}^T} & Z_{n-1}\Sigma_\epsilon (\varpi)   \cr\-
    \Sigma_\epsilon (\varpi) {Z_{n-1}^T}  &   \Sigma_\epsilon (\varpi)   \cr
 \end{pmat}} + {\sigma_{\epsilon n}^2(\varpi)}  {\begin{pmat}[{|}]
    Z_{n-1}\Sigma_\epsilon (\varpi)^{-1}  {Z_{n-1}^T} & -Z_{n-1}\Sigma_\epsilon (\varpi)^{-1}  \cr\-
    -\Sigma_\epsilon (\varpi)^{-1}  {Z_{n-1}^T} &   \Sigma_\epsilon (\varpi)^{-1}  \cr
 \end{pmat}} \right)   \\[4mm]
  \pmatset{1}{0.1pt}
  \pmatset{0}{0.1pt}
  \pmatset{2}{1pt}
  \pmatset{3}{1pt}
  \pmatset{4}{1pt}
  \pmatset{5}{1pt}
  \pmatset{6}{1pt}
\; + \left({\begin{pmat}[{|}]
    Z_{n-1}\Sigma_\epsilon (\varpi) {Z_{n-1}^T} & Z_{n-1} \Sigma_\epsilon (\varpi)   \cr\-
     \Sigma_\epsilon (\varpi) {Z_{n-1}^T} &    \Sigma_\epsilon (\varpi)  \cr
 \end{pmat}} + {\sigma_{\epsilon n}^2(\varpi)}
 {\begin{pmat}[{|}]
   {Z_{n-1}\Sigma_\epsilon (\varpi)^{-1}  {Z_{n-1}^T}} & -{Z_{n-1}\Sigma_\epsilon (\varpi)^{-1}  }  \cr\-
   -{\Sigma_\epsilon (\varpi)^{-1}  {Z_{n-1}^T}} & \Sigma_\epsilon (\varpi)^{-1}   \cr
 \end{pmat}}\right) {\begin{pmat}[{|}]
   {{A_{r\epsilon }}(\varpi )} & 0  \cr\-
   0 & {{A_\epsilon }(\varpi )}  \cr
 \end{pmat}^*}\\[4mm]
 \pmatset{1}{0.1pt}
  \pmatset{0}{0.1pt}
  \pmatset{2}{1pt}
  \pmatset{3}{1pt}
  \pmatset{4}{1pt}
  \pmatset{5}{1pt}
  \pmatset{6}{1pt}
  \; + {\begin{pmat}[{}]
 Z_{n-1}{B_\epsilon }(\varpi ) \cr\cr
 {B_\epsilon }(\varpi ) \cr
 \end{pmat}}{\begin{pmat}[{}]
 Z_{n-1}{B_\epsilon }(\varpi ) \cr \cr
 {B_\epsilon }(\varpi ) \cr
 \end{pmat}}^*+   {\sigma_{\epsilon n}^2(\varpi)} {\begin{pmat}[{}]
 Z_{n-1}\Sigma_\epsilon (\varpi)^{-1} {C_\epsilon }^*(\varpi ) \cr
\Sigma_\epsilon (\varpi)^{-1} {C^*}_\epsilon (\varpi ) \cr
 \end{pmat}}
{\begin{pmat}[{}]
 Z_{n-1}\Sigma_\epsilon (\varpi)^{-1} {C_\epsilon }^*(\varpi ) \cr
\Sigma_\epsilon (\varpi)^{-1} {C^*}_\epsilon (\varpi ) \cr
 \end{pmat}^*}
   \\[4mm]
  = 0 \\
 \end{array}
 \end{equation}

  \begin{equation}
  \small
  \label{ProofTheorem1_equ8}
  \begin{array}{*{20}{l}}
  \Delta_2\\
  =\pmatset{1}{0.2pt}
  \pmatset{0}{0.2pt}
  \pmatset{2}{1pt}
  \pmatset{3}{1pt}
  \pmatset{4}{1pt}
  \pmatset{5}{1pt}
  \pmatset{6}{1pt}    {\begin{pmat}[{|}]
        Z_{n-1}\Sigma_\epsilon (\varpi) Z_{n-1}^T   & Z_{n-1} \Sigma_ \epsilon (\varpi)   \cr\-
       \Sigma_\epsilon (\varpi) Z_{n-1}^T  &  \Sigma_\epsilon(\varpi)  \cr
  \end{pmat}}{\begin{pmat}[{}]
       -Z_{n-1}  \cr
         I \cr
  \end{pmat}}  \\
\;+\sigma^2_{\epsilon n}(\varpi) \pmatset{1}{0.2pt}
  \pmatset{0}{0.2pt}
  \pmatset{2}{1pt}
  \pmatset{3}{1pt}
  \pmatset{4}{1pt}
  \pmatset{5}{1pt}
  \pmatset{6}{1pt}
   {\begin{pmat}[{|}]
         Z_{n-1}\Sigma_\epsilon (\varpi)^{-1} Z_{n-1}^T  &  - Z_{n-1}\Sigma_\epsilon (\varpi)^{-1}   \cr\-
        -  \Sigma_\epsilon (\varpi)^{-1} Z_{n-1}^T    &  \Sigma_\epsilon (\varpi)^{-1}    \cr
  \end{pmat}}{\begin{pmat}[{}]
       -Z_{n-1}  \cr
         I \cr
  \end{pmat}}  +2 \sigma_{\epsilon n} (\varpi)  {\begin{pmat}[{}]
         \sigma_{\epsilon n}(\varpi)  Z_{n-1}\Sigma_\epsilon (\varpi)^{-1}  \cr
         - \sigma_{\epsilon n}(\varpi) \Sigma_\epsilon (\varpi)^{-1}   \cr
  \end{pmat}}   \\
  = 0 \\
\end{array}{\kern 46pt} \end{equation}

\begin{equation}
  \label{ProofTheorem1_equ9}
  \small
  \begin{array}{*{20}{l}}
\Delta_3\\
  =\pmatset{1}{0.2pt}
  \pmatset{0}{0.2pt}
  \pmatset{2}{1pt}
  \pmatset{3}{1pt}
  \pmatset{4}{1pt}
  \pmatset{5}{1pt}
  \pmatset{6}{1pt}
{\begin{pmat}[{|}]
        Z_{n-1}\Sigma_\epsilon (\varpi)Z_{n-1}^T   & Z_{n-1}\Sigma_\epsilon (\varpi)   \cr\-
       \Sigma_\epsilon (\varpi) Z_{n-1}^T  &  \Sigma_\epsilon (\varpi)  \cr
  \end{pmat}}  {\begin{pmat}[{}]
       -Z_{n-1}  \cr
         -I \cr
  \end{pmat}}+2 \sigma_{\epsilon n}(\varpi){\begin{pmat}[{}]
          \sigma_{\epsilon n}^{-1}(\varpi) Z_{n-1}\Sigma_\epsilon (\varpi)  \cr
          \sigma_{\epsilon n}^{-1}(\varpi) \Sigma_\epsilon (\varpi)   \cr
  \end{pmat}}   \\
 \;+\sigma^2_{\epsilon n}(\varpi)\pmatset{1}{0.2pt}
  \pmatset{0}{0.2pt}
  \pmatset{2}{1pt}
  \pmatset{3}{1pt}
  \pmatset{4}{1pt}
  \pmatset{5}{1pt}
  \pmatset{6}{1pt}
 {\begin{pmat}[{|}]
         Z_{n-1}\Sigma_\epsilon (\varpi)^{-1} Z_{n-1}^T  &  - Z_{n-1}\Sigma_\epsilon (\varpi)^{-1}   \cr\-
        - \Sigma_\epsilon (\varpi)^{-1}  Z_{n-1}^T    &   \Sigma_\epsilon (\varpi)^{-1}    \cr
  \end{pmat}} {\begin{pmat}[{}]
        -Z_{n-1}  \cr
         -I \cr
  \end{pmat}}       \\
  = 0 \\
\end{array} {\kern 114pt} \end{equation}

\noindent According to the Generalized KYP Lemma, the dilate error
systems $\mathcal E_n(\jmath\omega)$
 satisfying
 \[\sigma_{max} ( {\mathcal E_n(\jmath\varpi )} )  \le 2\sigma_{\epsilon n}(\varpi)  \]
therefore the error system $\mathcal E_n(\jmath\omega)$ satisfying
\[  \sigma_{max} ({\mathcal E_n(\jmath\varpi )} ) \le  \sigma_{max}({\mathcal E_n(\jmath\varpi )}) \le 2\sigma_{\epsilon n}(\varpi) \]
\noindent  This completes the SF-type error bound (\ref{SF-EB-SF}) for the $r=n-1$ case. The remainder of the proof for the $r=n-2,...1$ cases can be
easily completed in a reciprocal way.\\

\noindent 2). From (\ref{reducedmodel-SF}) and (\ref{Reduced_SF-FDES}), it can be concluded that the SF-type frequency-dependent extended system $G_{r\epsilon \varpi}(\jmath \omega)$ of reduced system $G_{r}(\jmath\omega)$ can be obtained by applying the standard FIBT algorithm for $G_{\epsilon \varpi}(\jmath \omega)$. Therefore, we have
\begin{equation}\begin{array}{l}
\sigma_{max}(G_{\epsilon \varpi}(\jmath \omega)  - G_{r \epsilon \varpi}(\jmath \omega)) \le 2 \sum\limits_{i = r + 1}^{n} {{\sigma _{i\varpi}}}, \; for \; all \; \omega \in [-\infty,+\infty] \\
 \end{array}\end{equation}
\noindent Noting that \begin{equation}
\small
\begin{array}{l}
G(\jmath \omega ) - {G_r}(\jmath \omega)
 = (G_{\epsilon\varpi}(\jmath \omega ) - G_{r\epsilon \varpi}(\jmath \omega)) +( G(\jmath \omega)-G_{\epsilon \varpi}(\jmath \omega) )+(G_{r \epsilon \varpi}(\jmath \omega)-G_r(\jmath \omega)). \\
 \end{array}\end{equation}
Using triangle inequality we get\begin{equation}
\small
\begin{array}{l}
\sigma_{max}(G(\jmath \omega ) - {G_r}(\jmath \omega))\\
\le  \sigma_{max}(G_{\epsilon\varpi}(\jmath \omega ) - G_{r\epsilon \varpi}(\jmath \omega))+\sigma_{max}(G_(\jmath \omega ) - G_{\epsilon \varpi}(\jmath \omega))+\sigma_{max}(G_{r}(\jmath \omega ) - G_{r\epsilon \varpi}(\jmath \omega))\\
  \le 2 \sum\limits_{i = r + 1}^{n} {{\sigma _{i\varpi}}} {\kern 1pt} +{\left\| G(\jmath \omega)-G_{\epsilon \varpi}(\jmath \omega) \right\|_\infty }+{\left\| G_r(\jmath \omega)-G_{r \epsilon \varpi}(\jmath \omega) \right\|_\infty }, {\kern 12pt} \; for \; all \; \omega \in [-\infty,+\infty] \\
 \end{array}\end{equation}
\noindent This completes the proof of entire-frequency error bound (\ref{EF-EB-SF}). 
\end{proof}

\noindent Based on above preliminaries and results, we now at the stage to present the SF-type frequency-dependent balanced truncation algorithm (see Algorithm 1).\\

\begin{spacing}{2.5}
\begin{algorithm}
\caption{SF-type FDBT}
\begin{algorithmic}
\REQUIRE {Full-order model $(A,B,C,D)$, frequency $(\varpi)$, user-defined parameter $\epsilon$ and the order of reduced model $(r)$}, \\[3mm]
  \textbf{Step 1.} Solve the SF-type frequency-dependent Lyapunov equations (\ref{SF-FDLE})\\ [3mm]

  \textbf{Step 2.}  Get the SF-type frequency-dependent balanced realization of the given system by coordinate transformation:
\begin{equation}
\small
\begin{array}{l}
 \pmatset{1}{0.36pt}
  \pmatset{0}{0.2pt}
  \pmatset{2}{4pt}
  \pmatset{3}{4pt}
  \pmatset{4}{4pt}
  \pmatset{5}{4pt}
  \pmatset{6}{4pt}    \begin{pmat}[{|}]
       A_{\epsilon b}(\varpi)   &   B_{\epsilon b}(\varpi)  \cr\-
       C_{\epsilon b}(\varpi)   &   D_{\epsilon b}(\varpi)  \cr
      \end{pmat}    =
  \pmatset{1}{0.36pt}
  \pmatset{0}{0.2pt}
  \pmatset{2}{4pt}
  \pmatset{3}{4pt}
  \pmatset{4}{4pt}
  \pmatset{5}{4pt}
  \pmatset{6}{4pt}
 \begin{pmat}[{|}]
      T_\epsilon^{-1}(\varpi) A_\epsilon(\varpi) T_\epsilon(\varpi)   & T_\epsilon^{-1}(\varpi)B_\epsilon(\varpi)  \cr\-
      C_\epsilon(\varpi) T_\epsilon(\varpi)    & D_\epsilon(\varpi) +C_{\epsilon b}(\varpi){(\epsilon I + \jmath \varpi I - A_{\epsilon b}(\varpi))^{ - 1}}B_{\epsilon b}(\varpi) \cr
      \end{pmat}     \\
 \end{array},\end{equation}

where $T_\epsilon(\varpi)$ is a matrix that simultaneously diagonalize the matrices $W_{c\epsilon}(\varpi)$ and $W_{o\epsilon}(\varpi)$, i.e.,

\[\small T^{-1}_\epsilon(\varpi)W_{c\epsilon}(\varpi)T_\epsilon(\varpi)=T_\epsilon^{*}(\varpi)W_{o\epsilon}(\varpi)T^{-*}_{\epsilon}(\varpi)=\Sigma_\epsilon(\varpi),\]
\textbf{Step 3.} Compute the reduced-order model as:
\begin{equation}
\label{truncatedsystemc}\begin{array}{l}
 A_r=\jmath \varpi I-\epsilon Z_r(\jmath \varpi I-A_{\epsilon b}(\varpi)){Z_r^T} (\epsilon I -Z_r(\jmath \varpi I-A_{\epsilon b}(\varpi)){Z_r^T})^{-1},\\
 B_r = \epsilon^{-1}(\epsilon  I + \jmath \varpi I - A_r)Z_rB_{\epsilon b}(\varpi), \\
 C_r = \epsilon^{-1} C_{\epsilon b}(\varpi)Z_r^T(\epsilon I + \jmath \varpi I - A_r),\\
 D_r =D_{\epsilon b}(\varpi) -C_r(\epsilon  I + \jmath \varpi I - A_r)^{-1}B_r. \\
 \end{array}\end{equation}
 \ENSURE Reduced-order model $(A_r,B_r,C_r,D_r)$
\end{algorithmic}
\end{algorithm}

\end{spacing}

\begin{remark} According to Proposition 3, the SF-type error bound can be regulated to an arbitrary small value by decreasing the parameter $\epsilon$, in other word, arbitrary approximation accuracy at the given frequency point $\omega=\varpi$ can be achieved. To make the approximation performance over the neighboring intervals ($\omega \in [\varpi-\delta, \varpi+\delta]$) be satisfactory, the value of parameter $\epsilon$ should be selected carefully. One possible way to pick an appropriate value of $\epsilon$ is to plot the curves of SF-type error bound (\ref{SF-EB-SF}) and EF-type error bound (\ref{EF-EB-SF}) with respect to the parameter $\epsilon$, then one can choose a proper value $\epsilon^*$ which make the SF-type and EF-type error bound be traded off against each other. Furthermore, it is suggested to adopt the value of $\epsilon$ be smaller than $\epsilon^*$ if there exists an estimation ($\hat \delta$) on the size of the uncertain frequency interval ($\delta$). The smaller $\hat \delta$ is, the smaller value of $\epsilon$ could be.
\end{remark}~\\

\begin{remark} For the sake of theoretical completeness, the SF-type FDBT approach is developed in a complex setting. The original system matrices and the reduced system matrices are allowed to be complex. In many applications, only real systems are of practical interest. With real model restriction, the proposed SF-type FDBT can only be applied in the case that $\varpi=0$. It is easy to find that the involved matrices $W_{c\epsilon}(\varpi), W_{o\epsilon}(\varpi), T_{\epsilon}(\varpi)$ and the generated reduced model $A_r,B_r,C_r,D_r$ are all real if the original system is real and the frequency point is $\varpi=0$. In the framework of balancing related methods, the proposed SF-type FDBT is not the only way for solving model order reduction problems assuming the dominating frequency is $\varpi=0$. As referred to in Section I, SPA is also regarded as an effective way for improving the approximation performance over low-frequency ranges. However, it should be noticed that the underlying mechanisms and the algorithms of SPA and SF-type FDBT are totally different. Which one will performs better on low-frequency approximation accuracy improvement depends on the given original system model. From the results of Example 3 in Section 5, to say the least, the proposed SF-type FDBT can be viewed as a new non-trivial alternative option besides SPA. 
\end{remark}~\\

\begin{remark} It is well-known that the conventional balanced truncation methods (such as the above mentioned FIBT, SPA, FWBT and FGBT) are developed for stable systems. To make those methods applicable for unstable system, some techniques like stable part and unstable part decomposition should be combined \cite{BT2_Survey} \cite{BT_unstable_Zhou} \cite{BT_unstable_Peter}. According to Proposition 2, one can always find a stable SF-type frequency-dependent extended system by choosing a proper $\epsilon$, even if the given original system is unstable. Thus, the SF-type FDBT can be used for coping with model reduction of unstable systems directly. The corresponding cost is that it cann't guarantee the generated reduced model is stable even if the original system is stable.\end{remark}~\\

\section{Frequency-dependent Balanced Truncation over Known Frequency-Intervals}

In this section, we present our results for the cases that the operating frequency belongs to a pre-known limited interval, i.e. $\omega \in [\varpi_1, \varpi_2]$. We will present some related definitions first and then show the related results and the interval-type frequency-dependent balanced truncation algorithm.\\

\begin{definition}[Interval-type Frequency-dependent Extend systems]\label{Int-FDES} Given a linear continuous-time system (\ref{originalsystem}) and a pre-known frequency interval ($\omega \in [\varpi_1,\varpi_2]$), one can construct an interval-type frequency-dependent extended system as follows: \begin{equation}
 \label{interval-FDES}
 G_{\varpi_1,\varpi_2}(\jmath \omega): \pmatset{1}{0.1pt}
  \pmatset{0}{0.1pt}
  \pmatset{2}{4pt}
  \pmatset{3}{4pt}
  \pmatset{4}{4pt}
  \pmatset{5}{4pt}
  \pmatset{6}{1pt}
  {\begin{pmat}[{|}]
       A (\varpi_1,\varpi_2) &  B(\varpi_1,\varpi_2) \cr\-
       C (\varpi_1,\varpi_2) &  D(\varpi_1,\varpi_2) \cr
  \end{pmat}},\end{equation}
\noindent  where  \[\begin{array}{l}
 A(\varpi_1,\varpi_2)=A,\\
 B(\varpi_1,\varpi_2)=[\varpi_d^2(\jmath\varpi _1I - A)^{ - 1}(\jmath\varpi _2I - A)^{ - 1}]^{\frac{1}{2}\star} B,\\
 C(\varpi_1,\varpi_2)=C [\varpi_d^2(\jmath\varpi _1I - A)^{ - 1}(\jmath\varpi _2I - A)^{ - 1}]^{\frac{1}{2}\star},\\
 D(\varpi_1,\varpi_2)=D+C[(\jmath \varpi_c I-A)(\jmath\varpi _1I - A)^{ - 1}(\jmath\varpi _2I - A)]^{ - 1} B,\\
 \varpi_d=(\varpi_2-\varpi_1)/2, {\kern 4pt} \varpi_c=(\varpi_2+\varpi_1)/2.\\
 \end{array}\]
\end{definition}~\\

\begin{definition}[Interval-type Frequency-dependent Lyapunov Equations] \label{def-Int-FDLE} Given a linear continuous-time system (\ref{originalsystem})
and a pre-specified frequency interval ($\omega \in [\varpi_1, \varpi_2]$), then the following
two Lyapunov equation 
\begin{equation}
\label{interval-FDLE}
\begin{array}{l}
A(\varpi_1,\varpi_2)W_c(\varpi_1,\varpi_2)+W_c(\varpi_1,\varpi_2)A^*(\varpi_1,\varpi_2)+B(\varpi_1,\varpi_2) B^*(\varpi_1,\varpi_2)=0\\
A^*(\varpi_1,\varpi_2)W_o(\varpi_1,\varpi_2)+W_o(\varpi_1,\varpi_2) A(\varpi_1,\varpi_2)+C^*(\varpi_1,\varpi_2) C(\varpi_1,\varpi_2)=0 \\
 \end{array}
 \end{equation}
\noindent  are defined as interval-type frequency-dependent controllability and
observability Lyapunov equations of the continuous-time system
(\ref{originalsystem}). Furthermore, the solutions $W_c(\varpi_1,\varpi_2)$ and $
W_o(\varpi_1,\varpi_2)$ will be referred to as interval-type frequency-dependent controllability and observability Gramians of the continuous-time system (\ref{originalsystem})
\end{definition}~\\

\begin{definition}[Interval-type Frequency-dependent Balanced Realization] \label{def-Int-FDBR} Given a linear continuous-time system (\ref{originalsystem})
and a pre-specified frequency interval ($\omega \in [\varpi_1, \varpi_2]$),  the corresponding interval-type frequency-dependent
controllability and observability Gramians are equal and diagonal, i.e. the following Lyapunov equations 
\begin{equation}
\label{Interval-FDLE-Balanced}\begin{array}{l}
A(\varpi_1,\varpi_2) \Sigma(\varpi_1,\varpi_2)+\Sigma(\varpi_1,\varpi_2) A^*_(\varpi_1,\varpi_2)+B(\varpi_1,\varpi_2)B^*(\varpi_1,\varpi_2)=0\\
A^*(\varpi_1,\varpi_2)\Sigma(\varpi_1,\varpi_2)+\Sigma(\varpi_1,\varpi_2)A(\varpi_1,\varpi_2)+C^*(\varpi_1,\varpi_2)C(\varpi_1,\varpi_2)=0 \\
 \end{array}\end{equation}
simultaneously hold, then this particular realization will be referred to as interval-type frequency-dependent balanced realization.
\end{definition}~\\

\begin{theorem}[Interval-type Frequency-dependent Balanced Truncation]\label{the-Int-FDBT} Given a linear continuous-time system (\ref{originalsystem})
with a pre-specified frequency interval ($\omega \in [\varpi_1, \varpi_2]$), and assume the system is given in interval-type frequency-dependent balanced
realization with respect to the interval-type frequency-dependent Gramian: \[\Sigma(\varpi_1,\varpi_2)=diag{(\sigma_{1}(\varpi_1,\varpi_2),...,\sigma_{r}(\varpi_1,\varpi_2), ...,\sigma_{n}(\varpi_1,\varpi_2))},\]\emph{and} $\sigma_{1}(\varpi_1,\varpi_2) \geq
...\geq\sigma_{r}(\varpi_1,\varpi_2)\geq...\geq\sigma_{n}(\varpi_1,\varpi_2),$
 the desired $r^{th}$ reduced-order model 
 $  \pmatset{1}{0.36pt}
  \pmatset{0}{0.2pt}
  \pmatset{2}{2pt}
  \pmatset{3}{2pt}
  \pmatset{4}{2pt}
  \pmatset{5}{2pt}
  \pmatset{6}{2pt}
G_r(\jmath \omega):= {   \begin{pmat}[{|}]
       A_r &  B_r \cr\-
      C_r    & D_r    \cr
  \end{pmat}}
$  is given by: 
\begin{equation}
\label{reducedmodel-Interval}\begin{array}{l}
 A_r =  Z_r A Z_r^T,\\
 B_r = [\varpi_d^2(\jmath\varpi _1I - A_r)^{ - 1}(\jmath\varpi _2I - A_r)^{ - 1}]^{-\frac{1}{2}\star} Z_r B(\varpi_1,\varpi_2), \\
 C_r = C(\varpi_1,\varpi_2) Z_r^T [\varpi_d^2(\jmath\varpi _1I - A_r)^{ - 1}(\jmath\varpi _2I - A_r)^{ - 1}]^{-\frac{1}{2}\star},\\
 {D_r} = D(\varpi_1,\varpi_2)-C_r[(\jmath \varpi_c I-A)(\jmath\varpi _1I - A_r)^{ - 1}(\jmath\varpi _2I - A_r)]^{ - 1}B_r, \\
 \end{array}\end{equation}

\noindent  where $
 {Z_r} =[I^{r \times r} {\kern 6pt}  0^{r \times (n - r)}]
$. Furthermore, the truncated model $G_r(\jmath \omega)$ possesses
the following properties: \\
 \noindent  1). If the original system is stable then the reduced system is stable.  \\
 \noindent  2). The approximation error between the original system model
(\ref{originalsystem}) and the truncated $r^{th}$ reduced model (\ref{reducedmodel-Interval})
over the given frequency interval ($\omega \in [\varpi_1,\varpi_2]$) satisfies the following interval-type error bound: 
\begin{equation}\small
\label{Interval-EB-Interval}
\sigma_{max}({G(j\omega ) - {G_r}(j\omega )}) \le  \sum\limits_{i = r + 1}^{n} {{\sqrt {\eta _{i}(\varpi_1,\varpi_2)}}}, {\kern 15pt} for {\kern 4pt}all {\kern 4pt} \omega \in  [\varpi_1, \varpi_2],
\end{equation}
 where 
   \begin{equation}
 \small \eta_i(\varpi_1,\varpi_2) = {\sigma _{max}}\left(  (2{\sigma _i(\varpi_1,\varpi_2)})^2 I +  He \left(-\mathscr C_{ei} \mathcal N_{ei} \mathscr B_{ei} He( [
   0 {\kern 6pt} I  ]^T (2{\sigma _i(\varpi_1,\varpi_2)}) [ I  {\kern 6pt} 0 ] \right) \right)
   \end{equation}
 and 
  \begin{equation}\label{dilatematrix1}
\scriptsize
\mathscr B_{ei}=\begin{array}{l}
  \pmatset{1}{0.36pt}
  \pmatset{0}{0.2pt}
  \pmatset{2}{8pt}
  \pmatset{3}{8pt}
  \pmatset{4}{8pt}
  \pmatset{5}{12pt}
  \pmatset{6}{8pt}    \begin{pmat}[{|}]
      \mathcal B_{ei}   &   \mathcal B_{di} \cr
           \end{pmat}    = \begin{pmat}[{|}]
      \mathcal M^{-1}_{ei} \left[ \begin{array}{l}
  Z_{i-1}\\
 \;\; Z_{i} \\
 \end{array} \right]  B(\varpi_1,\varpi_2)  &   \sigma_i(\varpi_1,\varpi_1) \mathcal M^{-1}_{ei} \Sigma_{ei}^{-1}(\varpi_1,\varpi_2) \left[ \begin{array}{l}
 \;\; Z_{i-1}\\
-Z_{i} \\
 \end{array} \right]  C^*(\varpi_1,\varpi_2)   \cr
           \end{pmat}
    \\
 \end{array},{\kern 42pt}\end{equation} \begin{equation}\label{dilatematrix2}
\scriptsize
\mathscr C_{ei}^*=\begin{array}{l}
  \pmatset{1}{0.36pt}
  \pmatset{0}{0.2pt}
  \pmatset{2}{8pt}
  \pmatset{3}{8pt}
  \pmatset{4}{8pt}
  \pmatset{5}{12pt}
  \pmatset{6}{8pt}    \begin{pmat}[{|}]
      \mathcal C_{ei}^*   &   \mathcal C_{di}^* \cr
           \end{pmat}
             = \begin{pmat}[{|}]
      \mathcal M_{ei}^{-*} \left[ \begin{array}{l}
 -Z_{i-1}\\
 \;\; Z_{i} \\
 \end{array} \right] C^*(\varpi_1,\varpi_2)
 &     \sigma_i(\varpi_1,\varpi_2) \mathcal M_{ei}^{-*} \Sigma_{ei}^{-1}(\varpi_1,\varpi_2) \left[ \begin{array}{l}
 - Z_{i-1}\\
 -Z_{i} \\
 \end{array} \right] B(\varpi_1,\varpi_2)    \cr
           \end{pmat}
    \\
 \end{array}, {\kern 48pt}\end{equation}
 \begin{equation}
\label{dilatematrix3}
\scriptsize
 \mathcal N_{ei} = diag\{N_{i-1}, N_{i}\}= diag\{ [(\jmath \varpi_c I-A_{i-1})(\jmath\varpi _1I - A_{i-1})^{ - 1}(\jmath\varpi _2I - A_{i-1})^{ - 1}], [(\jmath \varpi_c I-A_{i})(\jmath\varpi _1I - A_{i})^{ - 1}(\jmath\varpi _2I - A_{i})^{ - 1}]\},
\end{equation}\begin{equation}
\label{dilatematrix4}
\scriptsize
 \mathcal M_{ei}=diag\{M_{i-1},{\kern 4pt} M_{i}\}=  diag\{ [\varpi_d^2(\jmath\varpi _1I - A_{i-1})^{ - 1}(\jmath\varpi _2I - A_{i-1})^{-1}]^{\frac{1}{2}}, {\kern 4pt} [\varpi_d^2(\jmath\varpi _1I - A_{i})^{ - 1}(\jmath\varpi _2I - A_{i})^{-1}]^{\frac{1}{2}}\}, {\kern 44pt}
\end{equation}\begin{equation}
\label{dilatematrix5}
\scriptsize
 \Sigma_{ei}(\varpi_1,\varpi_2)=diag\{\Sigma_{i-1} (\varpi_1,\varpi_2),{\kern 4pt} \Sigma_{i}(\varpi_1,\varpi_2))\}= diag\{ Z_{i-1} \Sigma(\varpi_1,\varpi_2) Z^T_{i-1}, {\kern 4pt} Z_{i} \Sigma(\varpi_1,\varpi_2) Z^T_{i}\}. {\kern 100pt}
\end{equation}

\noindent  3). The approximation error between the original system model
(\ref{originalsystem}) and the truncated $r^{th}$ reduced model (\ref{reducedmodel-Interval})
over entire frequency range satisfies the following EF-type error bound: 
\begin{equation}
\label{EF-EB-Interval}
\begin{array}{l}
{ {\sigma_{max}(G(\jmath \omega ) - {G_r}(\jmath \omega)})
 } \le 2\sum\limits_{i = r + 1}^{n} {{\sigma _{i}(\varpi_1,\varpi_2)}} {\kern 22pt}\\
{\kern 114pt}+{\left\| G(\jmath \omega)-G_{\varpi_1,\varpi_2}(\jmath \omega) \right\|_\infty }\\
{\kern 114pt}+{\left\| G_r(\jmath \omega)-G_{r \varpi_1,\varpi_2}(\jmath \omega) \right\|_\infty }, {\kern 8pt} for {\kern 4pt} all {\kern 4pt} \omega \in [-\infty,+\infty]. \\
 \end{array}
 \end{equation}
\noindent  where $G_{r \varpi_1,\varpi_2}(\jmath \omega)$ represents the corresponding interval-type frequency-dependent extended system of reduced system $G_{r}(\jmath \omega)$, i.e. 
\begin{equation}
\small
 \label{interval-mapping}
 G_{r\varpi_1,\varpi_2}(\jmath \omega): \pmatset{1}{0.1pt}
  \pmatset{0}{0.1pt}
  \pmatset{2}{4pt}
  \pmatset{3}{4pt}
  \pmatset{4}{4pt}
  \pmatset{5}{4pt}
  \pmatset{6}{1pt}
  {\begin{pmat}[{|}]
       A_r(\varpi_1,\varpi_2) &  B_r(\varpi_1,\varpi_2) \cr\-
       C_r(\varpi_1,\varpi_2) &  D_r(\varpi_1,\varpi_2) \cr
  \end{pmat}},\end{equation}
\noindent \emph{where} \[\small\begin{array}{l}
 A_r(\varpi_1,\varpi_2)=A_r=Z_rA(\varpi_1,\varpi_2)Z_r^T,\\
 B_r(\varpi_1,\varpi_2)=[\varpi_d^2(\jmath\varpi _1I - A_r)^{ - 1}(\jmath\varpi _2I - A_r)^{ - 1}]^{\frac{1}{2}\star} B_r=Z_rB(\varpi_1,\varpi_2),\\
 C_r(\varpi_1,\varpi_2)=C_r [\varpi_d^2(\jmath\varpi _1I - A_r)^{ - 1}(\jmath\varpi _2I - A_r)^{ - 1}]^{\frac{1}{2}\star}=C(\varpi_1,\varpi_2)Z_r^T,\\
 D_r(\varpi_1,\varpi_2)=D_r+C_r[(\jmath \varpi_c I-A_r)(\jmath\varpi _1I - A_r)^{ - 1}(\jmath\varpi _2I - A_r)]^{ - 1} B_r=D(\varpi_1,\varpi_2).\\
 \end{array}\]
\end{theorem}

\begin{proof} 1) It can be easily completed by the similar procedure adopted in the proof of stability preservation for classic FIBT \cite{Zhoubook}. \\

\noindent 2). Similar with the proof of SF-type error bound provided in Theorem 1, only the sketch of the proof for $r=n-1$ case
will be given below. \\
\noindent We abuse notation a little bit for simplification. The error system $\mathcal E_n(\jmath\omega)$ between the original
system model $G(\jmath\omega)$ and the $(n-1)^{th}$ order reduced model $G_{n-1}(\jmath\omega)$ can be represented by:
\begin{equation}
\small
\begin{array}{l}
 \mathcal E_n(\jmath\omega)=G_n(\jmath\omega)-G_{n-1}(\jmath\omega)=G(\jmath\omega)-G_{n-1}(\jmath\omega)=: \pmatset{1}{0.36pt}
  \pmatset{0}{0.2pt}
  \pmatset{2}{8pt}
  \pmatset{3}{8pt}
  \pmatset{4}{8pt}
  \pmatset{5}{4pt}
  \pmatset{6}{4pt}    \begin{pmat}[{|}]
       \mathcal A_{en}   &  \mathcal B_{en} \cr\-
       \mathcal C_{en}   &   \mathcal  D_{en} \cr
      \end{pmat}    =
  \pmatset{1}{0.36pt}
  \pmatset{0}{0.2pt}
  \pmatset{2}{4pt}
  \pmatset{3}{4pt}
  \pmatset{4}{4pt}
  \pmatset{5}{2pt}
  \pmatset{6}{2pt}
 \begin{pmat}[{.|}]
      A_{n-1}   & 0  &   B_{n-1}\cr
      0     & A_n  &     B_{n}  \cr\-
       -C_{n-1} & C_n  &     D_n-D_{n-1}  \cr
      \end{pmat}     \\
 \end{array}\end{equation}

\noindent Based on the error system $\mathcal E_n(\jmath \omega)$,  one can construct a structure-preserving
dilated system $\mathcal E_n(\jmath \omega)$ as follows:
\begin{equation}
\small
  \begin{array}{l}
  \pmatset{1}{0.2pt}
  \pmatset{0}{0.2pt}
  \pmatset{2}{4pt}
  \pmatset{3}{4pt}
  \pmatset{4}{4pt}
  \pmatset{5}{4pt}
  \pmatset{6}{4pt}
\mathscr E_n(\jmath\omega):=\begin{pmat}[{|}]
      \mathscr A_{en}  & \mathscr B_{en}   \cr\-
      \mathscr C_{en}  & \mathscr D_{en}   \cr
      \end{pmat} =  \pmatset{1}{0.36pt}
  \pmatset{0}{0.2pt}
  \pmatset{2}{2pt}
  \pmatset{3}{2pt}
  \pmatset{4}{2pt}
  \pmatset{5}{2pt}
  \pmatset{6}{2pt}\begin{pmat}[{|.}]
       \mathcal A_{en}  &  \mathcal B_{en}                                                & \mathcal B_{dn}                                                         \cr\-
       \mathcal C_{en}  &  \mathcal D_{en}                                                & -\mathcal C_{en}\mathcal N_{en}\mathcal B_{dn}+2 \sigma _n(\varpi_1,\varpi_2) I            \cr
       \mathcal C_{dn}  & -\mathcal C_{dn}\mathcal N_{en}\mathcal B_{en}+2 \sigma _n(\varpi_1,\varpi_2) I   & -\mathcal C_{dn}\mathcal N_{en}\mathcal B_{dn}                                             \cr
      \end{pmat} \\
  \end{array}
\end{equation}

\noindent where $\mathcal B_{en}, \mathcal B_{dn}, \mathcal C_{en}, \mathcal C_{dn}, \mathcal N_{en}$ are defined as (\ref{dilatematrix1})-(\ref{dilatematrix5}). Now, if one choose two symmetrical Lyapunov variables $\mathscr Q_{en}=\mathscr Q_{en}^*\ge 0$ and $\mathscr P_{en}=\mathscr P_{en}^*$ as follows:
\begin{equation}
\label{PQ}
\small
\begin{array}{l}
\mathscr Q_{en} = {\mathcal N_{en}}({\omega _1},{\omega _2}){\mathcal B_{en}}{\mathcal B_{en}}^*{\mathcal N_{en}}^*({\omega _1},{\omega _2}) + {\mathcal N_{en}}({\omega _1},{\omega _2}){\mathcal B_{dn}}{\mathcal B_{dn}}^*{\mathcal N_{en}}^*({\omega _1},{\omega _2}) \\[6mm]
{\mathscr P_{en}} = He\left( {(\jmath{\omega _d}){{(\jmath{\omega _1}I - {\mathcal A_{en}})}^{ - 1}}{\mathcal B_{en}}{\mathcal B_{en}}^*{{(\jmath{\omega _2}I - {\mathcal A_{en}})}^{ - *}}} \right)\\[3mm]
 {\kern 16pt} + He\left( {(\jmath{\omega _d}){{(\jmath{\omega _1}I - {\mathcal A_{en}})}^{ - 1}}{\mathcal B_{dn}}{\mathcal B_{dn}}^*{{(\jmath{\omega _2}I - {\mathcal A_{en}})}^{ - *}}} \right) \\[3mm]
 {\kern 16pt}-\omega_d^2 He\left( {{{(\jmath{\omega _1}I - {\mathcal A_{en}})}^{ - 1}}{\mathcal M_{en}}^{ -1} [Z_{n-1}^T {\kern 4pt}I]^T \Sigma ({\omega _1},{\omega _2})[Z_{n-1}^T {\kern 4pt}I] {\mathcal M_{en}}^{ -*} {{(\jmath{\omega _2}I - {\mathcal A_{en}})}^{ - *}} } \right) \\[3mm]
  {\kern 16pt} - {\sigma _n}^2 \omega_d^2 He\left( {{{(\jmath{\omega _1}I - {\mathcal A_{en}})}^{ - 1}}{\mathcal M_{en}}^{ -1} [-Z_{n-1}^T {\kern 4pt}I]^T{\Sigma ^{ - 1}}({\omega _1},{\omega _2})[-Z_{n-1}^T {\kern 4pt}I]{\mathcal M_{en}}^{ -*} {{(\jmath{\omega _2}I - {\mathcal A_{en}})}^{ - *}} } \right) \\[3mm]
 \end{array}\end{equation}\noindent

\noindent  Combing the interval-type balanced frequency-dependent Lyapunov equation (\ref{Interval-FDLE-Balanced}) and following a similar way in the proof of Theorem 1, one can derive the inequality  \begin{equation}
\footnotesize \begin{array}{l}
 {\kern 11pt}\pmatset{1}{0.1pt}
  \pmatset{0}{0.1pt}
  \pmatset{2}{4pt}
  \pmatset{3}{4pt}
  \pmatset{4}{4pt}
  \pmatset{5}{4pt}
  \pmatset{6}{1pt}
  {\begin{pmat}[{|}]
    -He((\jmath \varpi_1 I-\mathscr A_{en})) \mathscr Q_{en} (\jmath \varpi_2-\mathscr A_{en})^*)+ \mathscr A_{en} \mathscr P_{en} + \mathscr P_{en}\mathscr A_{en}^*  + \mathscr B_{en} {\mathscr B_{en}^*}  & (\jmath \varpi_c I -\mathscr A_{en}) \mathscr C_{en}^*+ \mathscr P_{en}{\mathscr C_{en}^*} + \mathscr B_{en}{\mathscr D_{en}^*} \cr\-
      *  & -\mathscr C_{en}\mathscr Q_{en}\mathscr C_{en}^*+ \mathscr D_{en} {\mathscr D_{en}^*} - {(\sqrt{\eta_n})^2}I  \cr
  \end{pmat}} \\[4mm]
 = \pmatset{1}{0.1pt}
  \pmatset{0}{0.1pt}
  \pmatset{2}{4pt}
  \pmatset{3}{4pt}
  \pmatset{4}{4pt}
  \pmatset{5}{4pt}
  \pmatset{6}{1pt}{\begin{pmat}[{|}]
   {\kern 30pt}  0 {\kern 30pt}& {\kern 30pt}0 {\kern 30pt} \cr\-
   {\kern 30pt}  * {\kern 30pt} & (2{\sigma _n(\varpi_1,\varpi_2)})^2 I +  He \left( -\mathscr C_{ei} \mathcal N_{ei} \mathscr B_{ei} He( [
   0 {\kern 6pt} I ]^T (2{\sigma _n(\varpi_1,\varpi_2)})[  I {\kern 6pt} 0] ) \right)  - \eta_n(\varpi_1,\varpi_2) I  \cr
  \end{pmat}} \le 0 \\
 \end{array}\end{equation}

\noindent  According to the Generalized KYP Lemma, the dilated error system $\mathscr E_n(\jmath \omega)$
 satisfies \\[-8mm] \begin{equation}
\small
  \sigma_{max}(\mathscr E_n(\jmath\omega ))   \le \sqrt {\eta _n(\varpi_1,\varpi_2)}, {\kern 15pt} for {\kern 4pt} all {\kern 4pt} \omega \in  [\varpi_1, \varpi_2]
\end{equation}
\\[-12mm] Therefore the error system satisfying the following inequality\\[-8mm]
\begin{equation}
\small
  \sigma_{max}( \mathcal E_n(\jmath\omega )) \le \sigma_{max}(\mathscr E_n(\jmath\omega ))  \le \sqrt {\eta _n(\varpi_1,\varpi_2)}, {\kern 15pt} for {\kern 4pt} all {\kern 4pt} \omega \in  [\varpi_1, \varpi_2]
\end{equation}
\noindent This completes the proof of interval-type error bound (\ref{Interval-EB-Interval}) for the $r=n-1$ case, the $r=n-2,...,1$ cases can be fulfilled step by step.\\

\noindent 3). Similar with proof of EF-type error bound (\ref{EF-EB-SF}) provided by SF-type FDBT, the proof of EF-type error bound (\ref{EF-EB-Interval}) provided by interval-type FDBT can be completed in the same way.
\end{proof}~\\

\begin{proposition} \emph{the the following statements are true:}\\
{\emph{a).} $\mathop {\lim }\limits_{\varpi_d \to 0} W_c(\varpi_1,\varpi_2) = 0, {\kern 2pt} \mathop {\lim }\limits_{\varpi_d  \to 0}
W_o(\varpi_1,\varpi_2) = 0, {\kern 2pt} \mathop {\lim }\limits_{\varpi_d
\to 0} {\Sigma_\varpi(\varpi_1,\varpi_2)
} = 0 $,\\
\emph{b)} $\mathop {\lim }\limits_{\varpi_d  \to \infty}
W_c(\varpi_1,\varpi_2) = W_c, {\kern 2pt} \mathop {\lim }\limits_{\varpi_d \to \infty} W_o(\varpi_1,\varpi_2) = W_o, {\kern 2pt} \mathop {\lim
}\limits_{\varpi_d \to \infty} {\Sigma_\varpi (\varpi_1,\varpi_2)} = \Sigma$.} \\
\emph{c)} $\mathop {\lim }\limits_{\varpi_d  \to 0}
\eta_i=0,  {\kern 4pt}   i=1,...,n$ 
\end{proposition}
\begin{proof} a). It can be easily observed that
\begin{equation}
\small
\begin{array}{l} \mathop {\lim }\limits_{\varpi_d \to
0} {A(\varpi_1,\varpi_2)} =\mathop {\lim }\limits_{\varpi_d \to
0 } {A}=
A,\\
\mathop {\lim }\limits_{\varpi_d \to 0} {B(\varpi_1,\varpi_2)}
=\mathop {\lim }\limits_{\varpi_d  \to 0} {[\varpi_d^2(\jmath\varpi _1I - A)^{ - 1}(\jmath\varpi _2I - A)^{ - 1}]^{\frac{1}{2}\star} B}=0,\\
\mathop {\lim }\limits_{\varpi_d \to 0} {C(\varpi_1,\varpi_2)}
=\mathop {\lim }\limits_{\varpi_d \to 0 } { C [\varpi_d^2(\jmath\varpi _1I - A)^{ - 1}(\jmath\varpi _2I - A)^{ - 1}]^{\frac{1}{2}\star} } =0.\\
 \end{array}\end{equation}
\noindent From the interval-type frequency-dependent Lyapunov equation (\ref{interval-FDLE}), we know that
 \begin{equation}
\small
\label{lyapunov_equation}\begin{array}{l}
A\mathop {\lim }\limits_{\varpi_d \to 0} W_c(\varpi_1,\varpi_2) +\mathop {\lim }\limits_{\varpi_d \to 0} W_c(\varpi_1,\varpi_2) A^*=0\\
A^*\mathop {\lim }\limits_{\varpi_d \to 0} W_o(\varpi_1,\varpi_2) +\mathop {\lim }\limits_{\varpi_d \to 0} W_o(\varpi_1,\varpi_2) A=0. \\
 \end{array}\end{equation}
\noindent which means $\mathop {\lim }\limits_{\varpi_d \to 0} W_c(\varpi_1,\varpi_2) = 0$ and $\mathop {\lim }\limits_{\varpi_d  \to 0}
W_o(\varpi_1,\varpi_2) = 0$.

\noindent b). Similar with the above proof, we have
\begin{equation}
\small
\begin{array}{l} \mathop {\lim }\limits_{\varpi_d \to
\infty} {A(\varpi_1,\varpi_2)} =\mathop {\lim }\limits_{\varpi_d \to
\infty } {A}=
A,\\
\mathop {\lim }\limits_{\varpi_d \to \infty} {B(\varpi_1,\varpi_2)}
=\mathop {\lim }\limits_{\varpi_d  \to \infty} {[\varpi_d^2(\jmath\varpi _1I - A)^{ - 1}(\jmath\varpi _2I - A)^{ - 1}]^{\frac{1}{2}\star} B}=B,\\
\mathop {\lim }\limits_{\varpi_d \to \infty} {C(\varpi_1,\varpi_2)}
=\mathop {\lim }\limits_{\varpi_d \to \infty} { C [\varpi_d^2(\jmath\varpi _1I - A)^{ - 1}(\jmath\varpi _2I - A)^{ - 1}]^{\frac{1}{2}\star} } =C.\\
 \end{array}\end{equation}

\noindent and
 \begin{equation}
\small
\label{lyapunov_equation}\begin{array}{l}
A\mathop {\lim }\limits_{\varpi_d \to \infty} W_c(\varpi_1,\varpi_2) +\mathop {\lim }\limits_{\varpi_d \to \infty} W_c(\varpi_1,\varpi_2) A^*+BB^*=0\\
A^*\mathop {\lim }\limits_{\varpi_d \to \infty} W_o(\varpi_1,\varpi_2) +\mathop {\lim }\limits_{\varpi_d \to \infty} W_o(\varpi_1,\varpi_2) A+C^*C=0. \\
 \end{array}\end{equation}
\noindent Then $\mathop {\lim }\limits_{\varpi_d  \to \infty}
W_c(\varpi_1,\varpi_2) = W_c$ and $\mathop {\lim }\limits_{\varpi_d \to \infty} W_o(\varpi_1,\varpi_2) = W_o$ can be conclude.

\noindent c). Noticing that $\sigma_i(\varpi_1,\varpi_2)$ is the minimum of the diagonal components of $\Sigma_{ei}(\varpi_1,\varpi_2)$, then we have
 \begin{equation}
\small
\label{lyapunov_equation}\begin{array}{l}
\mathop {\lim }\limits_{\varpi_d \to 0} \sigma_i(\varpi_1,\varpi_2) \Sigma_{ei}(\varpi_1,\varpi_2) \leq I\\
 \end{array}\end{equation}
furthermore, one can conclude that there exists a scalar $\mu<\infty$ such that the following inequality

  \begin{equation}
\small
\label{lyapunov_equation}\begin{array}{l}
\mathop {\lim }\limits_{\varpi_d \to 0}  {\mathscr C_{ei} \mathcal N_{ei} \mathscr B_{ei}}  \leq \mu I\\
 \end{array}\end{equation}
\noindent holds since the convergence of matrices $\mathscr C_{ei}$,$N_{ei}$ and $\mathscr B_{ei}$ in cases that $\varpi_d \to 0$ are norm bounded.

  \begin{equation}
\small
\label{lyapunov_equation}\begin{array}{l}
{\kern 8pt} \mathop {\lim }\limits_{\varpi_d \to 0} \eta_i(\varpi_1,\varpi_2)\\
= \mathop {\lim }\limits_{\varpi_d \to 0} \sigma_i^2(\varpi_1,\varpi_2) I+  \mathop {\lim }\limits_{\varpi_d \to 0} He \left(-\mathscr C_{ei} \mathcal N_{ei} \mathscr B_{ei} He( [
   0 {\kern 6pt} I  ]^T (2{\sigma _i(\varpi_1,\varpi_2)}) [ I  {\kern 6pt} 0 ] \right)  \\
= \mathop {\lim }\limits_{\varpi_d \to 0} \sigma_i^2(\varpi_1,\varpi_2) I+  He \left(-\mathop {\lim }\limits_{\varpi_d \to 0} \mathscr C_{ei} \mathcal N_{ei} \mathscr B_{ei}  He( [
   0 {\kern 6pt} I  ]^T (2 \mathop {\lim }\limits_{\varpi_d \to 0} {\sigma _i(\varpi_1,\varpi_2)}) [ I  {\kern 6pt} 0 ] \right)  \\
   =0
 \end{array}\end{equation}

\noindent Thus the proof is completed.\end{proof}

\begin{proposition} The following equation 
\begin{equation}
\label{pro5}
T^{-1}[\varpi_d^2(\jmath\varpi _1I - A)^{ - 1}(\jmath\varpi _2I - A)^{ - 1}]^{\frac{1}{2}\star}T=[\varpi_d^2(\jmath\varpi _1I - T^{-1}AT)^{ - 1}(\jmath\varpi _2I - T^{-1}AT)^{ - 1}]^{\frac{1}{2}\star
}\end{equation} holds for arbitrarily given invertible matrix $T\in \mathbb C^{n\times n}$ 
\end{proposition}
\begin{proof}
Lets consider the square of matrices of the left side and right side in (\ref{pro5}), we have
\[\begin{array}{l}
\left(T^{-1}[\varpi_d^2(\jmath\varpi _1I - A)^{ - 1}(\jmath\varpi _2I - A)^{ - 1}]^{\frac{1}{2}\star}T\right)^2\\
= T^{-1}[\varpi_d^2(\jmath\varpi _1I - A)^{ - 1}(\jmath\varpi _2I - A)^{ - 1}]^{\frac{1}{2}\star}[\varpi_d^2(\jmath\varpi _1I - A)^{ - 1}(\jmath\varpi _2I - A)^{ - 1}]^{\frac{1}{2}\star}T\\
=T^{-1}[\varpi_d^2(\jmath\varpi _1I - A)^{ - 1}(\jmath\varpi _2I - A)^{ - 1}]T\\
=[\varpi_d^2(\jmath\varpi _1I - T^{-1}AT)^{ - 1}(\jmath\varpi _2I - T^{-1}AT)^{ - 1}]\\
=\left([\varpi_d^2(\jmath\varpi _1I - T^{-1}AT)^{ - 1}(\jmath\varpi _2I - T^{-1}AT)^{ - 1}]^{\frac{1}{2}\star}\right)^2\\
 \end{array}\]

\noindent The above equation means that there exist matrices $U,V$ such that
\[\begin{array}{l}
\left(T^{-1}[\varpi_d^2(\jmath\varpi _1I - A)^{ - 1}(\jmath\varpi _2I - A)^{ - 1}]^{\frac{1}{2}\star}T\right)^2\\
=\left([\varpi_d^2(\jmath\varpi _1I - \hat A)^{ - 1}(\jmath\varpi _2I - \hat A)^{ - 1}]^{\frac{1}{2}\star}\right)^2\\
=U V U^{-1}\\
 \end{array}\]
\noindent where $U$ is the matrix whose columns are eigenvectors of $\left(T^{-1}[\varpi_d^2(\jmath\varpi _1I - A)^{ - 1}(\jmath\varpi _2I - A)^{ - 1}]^{\frac{1}{2}\star}T\right)^2$ and $\left([\varpi_d^2(\jmath\varpi _1I - \hat A)^{ - 1}(\jmath\varpi _2I - \hat A)^{ - 1}]^{\frac{1}{2}\star}\right)^2$ and $V$ is the diagonal matrix whose diagonal elements are the corresponding eigenvalues. Furthermore, one get
\[\begin{array}{l}
T^{-1}[\varpi_d^2(\jmath\varpi _1I - A)^{ - 1}(\jmath\varpi _2I - A)^{ - 1}]^{\frac{1}{2}\star}T=[\varpi_d^2(\jmath\varpi _1I - \hat A)^{ - 1}(\jmath\varpi _2I - \hat A)^{ - 1}]^{\frac{1}{2}\star}=U V^{\frac{1}{2}\star} U^{-1}\\
 \end{array}\]
\noindent
 This completes the proof. \end{proof}~\\

 \noindent With the above preparations, the corresponding interval-type FDBT algorithm (Algorithm 2) can be presented as follows.\\

\begin{algorithm}
\caption{Interval-type FDBT}
\begin{algorithmic}
\REQUIRE{Full-order model $(A,B,C,D)$, Frequency interval $(\varpi_1, \varpi_2)$, order of reduced  model $(r)$.}\\[4mm]

\textbf{Step 1.} Solve the interval-type frequency-dependent Lyapunov equations (\ref{interval-FDLE})\\[4mm]

\textbf{Step 2.} Get the frequency-dependent realization of the given system by coordinate transformation:\\
\begin{small}
\begin{equation}
\begin{array}{l}
 \pmatset{1}{0.36pt}
  \pmatset{0}{0.2pt}
  \pmatset{2}{4pt}
  \pmatset{3}{4pt}
  \pmatset{4}{4pt}
  \pmatset{5}{4pt}
  \pmatset{6}{4pt}    \begin{pmat}[{|}]
       A_b   &   B_b \cr\-
       C_b   &   D_b \cr
      \end{pmat}    =
  \pmatset{1}{0.36pt}
  \pmatset{0}{0.2pt}
  \pmatset{2}{4pt}
  \pmatset{3}{4pt}
  \pmatset{4}{4pt}
  \pmatset{5}{4pt}
  \pmatset{6}{4pt}
 \begin{pmat}[{|}]
      T^{-1}(\varpi_1,\varpi_2) A T(\varpi_1,\varpi_2)   & T^{-1}(\varpi_1,\varpi_2)B  \cr\-
      C T(\varpi_1,\varpi_2)    & D+C_b[(\jmath\varpi _cI - A_b)(\jmath\varpi _1I - A_b)^{ - 1}(\jmath\varpi _2I - A_b)^{ - 1}]B_b \cr
      \end{pmat}     \\
 \end{array},
 \end{equation}
\end{small}

where $T(\varpi_1,\varpi_2)$ is a matrix that simultaneously diagonalize the matrices $W_c(\varpi_1,\varpi_2)$ and $W_o(\varpi_1,\varpi_2)$, i.e.,
\begin{small}\[ T^{-1}(\varpi_1,\varpi_2)W_c(\varpi_1,\varpi_2)T(\varpi_1,\varpi_2)=T^{*}(\varpi_1,\varpi_2)W_o(\varpi_1,\varpi_2)T^{-*}(\varpi_1,\varpi_2)=\Sigma(\varpi_1,\varpi_2),\]
\end{small}
\textbf{Step 3.} Compute the reduced-order model as:
\begin{equation}
\small
\label{RedMod-Interval}
\begin{array}{l}
 A_r =  Z_r A_b Z_r^T,\\[2mm]
 B_r = [\varpi_d^2(\jmath\varpi _1I - A_r)^{ - 1}(\jmath\varpi _2I - A_r)^{ - 1}]^{-\frac{1}{2}\star} Z_r [\varpi_d^2(\jmath\varpi _1I - A_b)^{ - 1}(\jmath\varpi _2I - A_b)^{ - 1}]^{\frac{1}{2}\star} B_b, \\[2mm]
 C_r = C_b[\varpi_d^2(\jmath\varpi _1I - A_b)^{ - 1}(\jmath\varpi _2I - A_b)^{ - 1}]^{\frac{1}{2}\star} Z_r^T [\varpi_d^2(\jmath\varpi _1I - A_r)^{ - 1}(\jmath\varpi _2I - A_r)^{ - 1}]^{-\frac{1}{2}\star},\\[2mm]
 {D_r} = D+C_b[(\jmath\varpi _cI - A_b)(\jmath\varpi _1I - A_b)^{ - 1}(\jmath\varpi _2I - A_b)^{ - 1}]B_b\\
 {\kern 36pt}-C_r[(\jmath\varpi _cI - A_r)(\jmath\varpi _1I - A_r)^{ - 1}(\jmath\varpi _2I - A_r)^{ - 1}]B_r, \\
 \end{array}
 \end{equation}

\ENSURE Reduced-order model $(A_r,B_r,C_r,D_r)$
\end{algorithmic}
\end{algorithm}

\begin{remark} Compared with other balancing-related approaches, the most distinctive feature of the proposed interval-type FDBT method is that it gives an interval-type error bound (\ref{Interval-EB-Interval}). To the best of our knowledge, it is the first time to provide such an interval-type error bound using the interval-type index (\ref{FF_index}) in the model order reduction research areas. In particular, as revealed by Proposition 4, the interval-type error bound (\ref{Interval-EB-Interval}) always tends to be zero while the interval size tends to zero. This property means that the interval-type FDBT generally will gives rise to good in-band approximation performance while provides better in-band error bound simultaneously as long as the size of frequency interval is small enough. Although the interval-type error bound may be increasing quickly with respect to the size of frequency interval. The interval-type error bound and its property are still appealing from a theoretical viewpoint.\end{remark}

\begin{remark} Again, the interval-type FDBT is also presented in a general form, i.e. the system matrices are allowed to be complex or real and the frequency interval might be symmetrical or asymmetrical. It can be easily verified that the interval-type FDBT will generate real reduced models for real full models if the given frequency interval is symmetrical (i.e $\varpi_1=-\varpi_2$). For applications with real system parameter restriction in asymmetrical frequency interval cases ($\omega \in [\varpi_1,\varpi_2]$), the interval-type FDBT can also be applied in a conservative way by modifying the frequency as $\omega \in [-\varpi_{max},\varpi_{max}]$ with $\varpi_{max}=max\{\left|\varpi_1\right|, \left|\varpi_2\right| \}$. 
\end{remark}

\section{Examples}

\begin{example} Lets consider a LTI system (\ref{originalsystem}) with the following parameter matrices:\begin{equation}
\small
\begin{array}{l}
 \pmatset{1}{0.36pt}
  \pmatset{0}{0.2pt}
  \pmatset{2}{4pt}
  \pmatset{3}{4pt}
  \pmatset{4}{10pt}
  \pmatset{5}{10pt}
  \pmatset{6}{10pt}    \begin{pmat}[{|}]
       A    &   B \cr\-
       C    &   D   \cr
      \end{pmat}    =
  \pmatset{1}{0.36pt}
  \pmatset{0}{0.2pt}
  \pmatset{2}{2pt}
  \pmatset{3}{2pt}
  \pmatset{4}{2pt}
  \pmatset{5}{2pt}
  \pmatset{6}{2pt}
 \begin{pmat}[{.....|}]
    0.2128  &  0.7749  &  0.1945  & -0.2864  &  0.0501 &  -0.0464  & 0.9673 \cr
   -0.6613  & -2.6801  & -0.8468  & -0.5733  & -0.7945 &   0.9653  & -1.4467 \cr
    0.2423  & -0.8043  & -0.7669  & -0.5423  & -0.9032 &   0.1441  & -1.2514\cr
   -0.1508  &  0.5229  &  0.6927  & -0.0704  &  0.8778 &  -0.5350  & -0.4141\cr
    0.3542  &  0.7882  &  0.3681  & -0.2077  & -0.1705 &  -0.7660  & -0.6560\cr
   -0.6424  & -0.5045  & -0.0252  &  0.6453  &  0.9838 &  -0.9392  & -0.1651\cr\-
  -1.5883   & -1.3181  &  0.5656  &  1.1507  & -0.5106 &  -0.7736  & 3.9764\cr
      \end{pmat}
 \end{array}\end{equation}

\noindent Here we assume that the frequency of input signal belongs to an uncertain interval around $\varpi=0$. The task is to build reduced model of order 3 approximating the frequency domain dynamic behaviors of the original model well in the neighborhood of $\varpi=0$. Among the existing balancing-related methods, the \emph{(generalized) SPA} is the most suitable one for coping with this kind of model reduction problems. At the same time, our proposed \emph{SF-type FDBT} method can also be applied for this kind of problems. The sigma plots of error systems generated by \emph{generalized SPA} and \emph{SF-type FDBT} are depicted in Fig.1 and Fig.2, respectively. As Fig.1 and Fig.2 shown, both of them could gives rise to small approximation error around $\varpi=0$. Moreover, one can make a tradeoff between the local approximation performance and global approximation performance by adjusting the the user-defined parameter ($\rho$ for generalized SPA and $\epsilon$ for SF-type FDBT). In this example the generalized SPA and the SF-type FDBT performs very similar with each other, however, huge variety on their performance  may occurred in some cases (see example 3 in the below, in which only the SF-type FDBT is effective).

   \begin{figure}[ht!]
    \centering
      \includegraphics[scale=0.45]{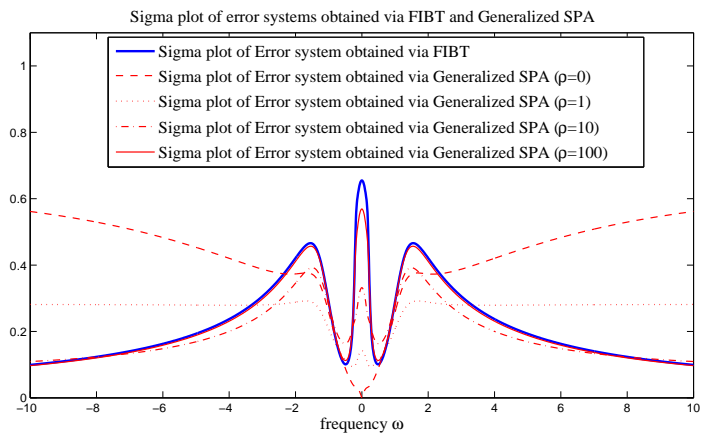}
      \vspace{-0.0cm}
    \caption{Sigma plot of error models generated via Generalized SPA and FIBT}
       \end{figure}
   \begin{figure}[ht!]
    \centering
      \includegraphics[scale=0.45]{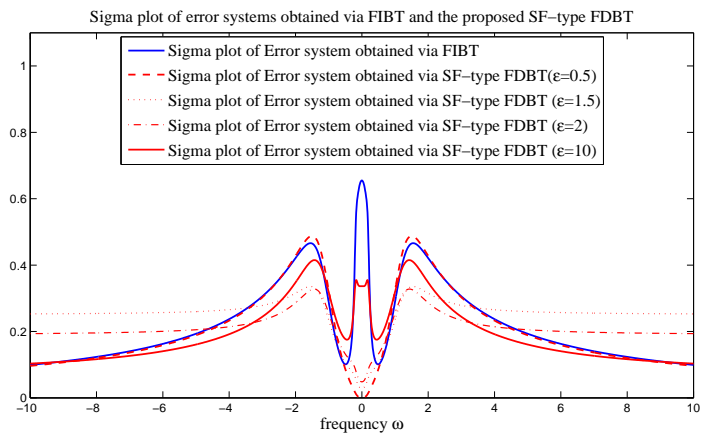}
      \vspace{-0.0cm}
    \caption{Sigma plot of error models generated via SF-type FDBT and FIBT}
      \end{figure}
   \begin{figure}[ht!]
    \centering
      \includegraphics[scale=0.45]{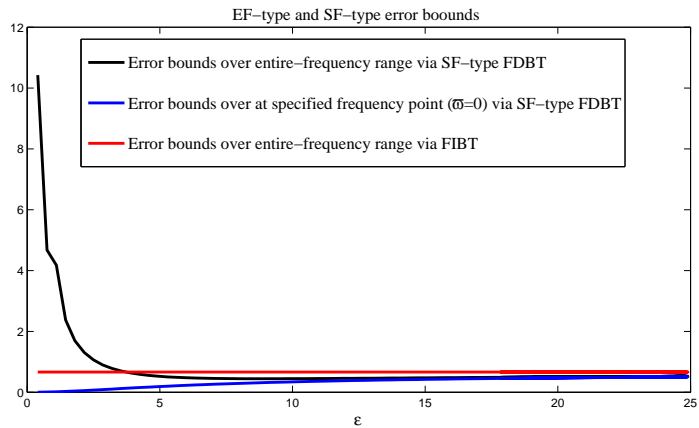}
      \vspace{-0.0cm}
    \caption{SF-type error bound and EF-type error bound with respect to the parameter $\epsilon$}
       \end{figure}

\noindent Besides, the corresponding SF-type error bound and EF-type error bound with respect to different $\epsilon$ provided by SF-type FDBT are plotted in Fig 3. According to the error bounds, we know that the local and global approximation performance could be well balanced by picking up the value of parameter $\epsilon$ larger than 3 and smaller than 5. In this way the trial-and-error procedure to find an appropriate $\epsilon$ can be shorten or avoided. Furthermore, if the parameter $\epsilon$ satisfy $25>\epsilon>4$, the EF-type error bound of SF-type FDBT will even be smaller that the EF-type error bound of FIBT.
\end{example}~\\

\begin{example} Lets consider a LTI system (\ref{originalsystem}) with the following parameter matrices:
\begin{equation}
\small
\begin{array}{l}
 \pmatset{1}{0.36pt}
  \pmatset{0}{0.2pt}
  \pmatset{2}{4pt}
  \pmatset{3}{4pt}
  \pmatset{4}{10pt}
  \pmatset{5}{10pt}
  \pmatset{6}{10pt}    \begin{pmat}[{|}]
       A    &   B \cr\-
       C    &   D   \cr
      \end{pmat}    =
  \pmatset{1}{0.36pt}
  \pmatset{0}{0.2pt}
  \pmatset{2}{2pt}
  \pmatset{3}{2pt}
  \pmatset{4}{2pt}
  \pmatset{5}{2pt}
  \pmatset{6}{2pt}
 \begin{pmat}[{...|}]
   -0.62  &  0.44 &  -0.03 &  -0.00 & -0.31 \cr
    0.44  & -3.64 &   0.59 &   0.02 &  0.47 \cr
    0.03  & -0.59 &  -6.80 &  -0.46 &  0.12 \cr
   -0.00  &  0.02 &   0.46 &  -5.64 & -0.00 \cr\-
   -0.31  &  0.47 &  -0.12 &  -0.00 &  0.00 \cr
      \end{pmat}
 \end{array}\end{equation}
\noindent  The frequency range of input signals is assumed to be pre-known, and we consider the following two different cases: (1) Case 1: $\omega \in [-0.4, +0.4]$; (2) Case 2: $\omega \in [-0.8, +0.8]$. \\
\noindent Among the existing balancing-related methods, FGBT \cite{BT2_Survey}  is the exact one developed for solving such interval-type finite-frequency model reduction problems. Our proposed interval-type FDBT is also aimed to solve this kind of problems. We will show the differences between them by this example. The sigma plot of error models and the corresponding error bound are given in the Fig. 4-Fig. 5, by which the most striking difference on the type of error bounds can be illustrated. The FGBT provides error bound over entire-frequency range, in contrast, the interval-type only provides error bound over the pre-specified frequency interval. Since it is assumed that the operating frequencies belong to the given intervals, the interval-type error bounds are adequate for approximation performance estimation. Compared with the standard FIBT, both the FGBT and the interval-type FDBT are effective in improving the approximation performance over specified frequency interval. At the same time, the interval-type FDBT has the advantage that it gives rise to better approximation performance and smaller error bound simultaneously.

   \begin{figure}[ht!]
    \centering
      \includegraphics[scale=0.6]{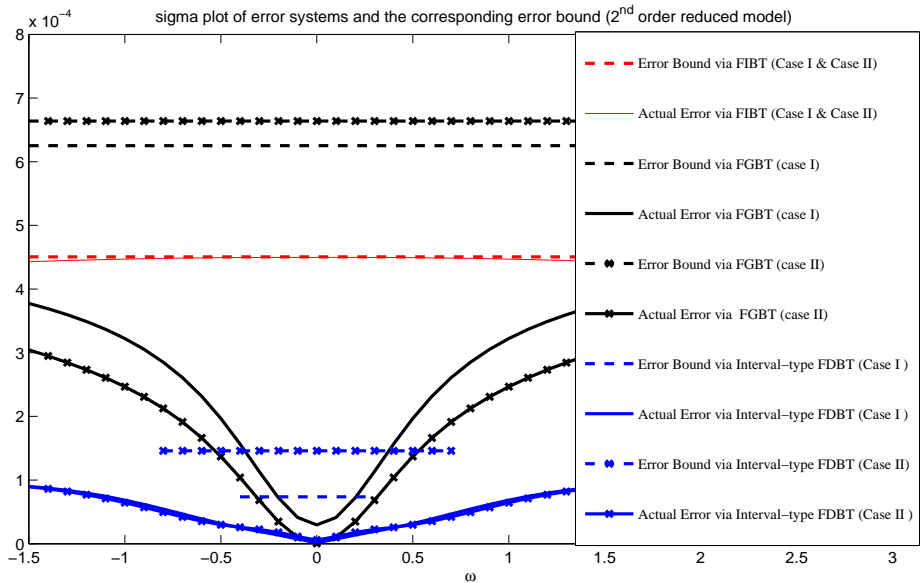}
         \vspace{-0.0cm}
          \caption{Sigma plot of error models and the corresponding error bounds ($2^{nd}$ order reduced model)}
       \end{figure}\vspace{-0.0cm}

   \begin{figure}[ht!]
    \centering
      \includegraphics[scale=0.6]{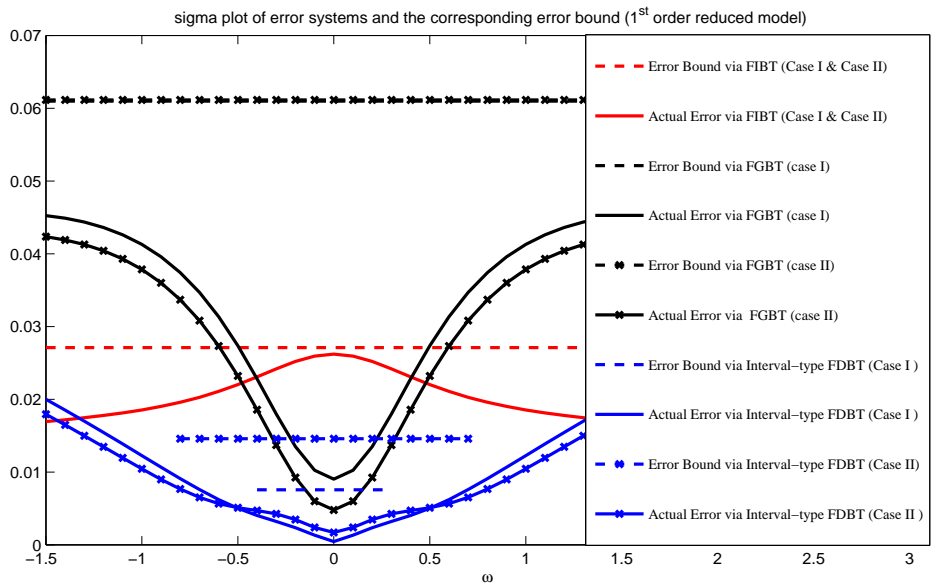}
       \vspace{-0.0cm}
       \caption{Sigma plot of error models and the corresponding error bounds ($1^{st}$ order reduced model)}
       \end{figure}

\noindent As depicted by Fig. 4 and Fig.5, the interval-type error bound provided by interval-type FDBT for Case II is larger than the interval-type error bound for Case I. To further show the property of interval-type error bound, we plot the curves of the two interval-size ($\varpi_l$) dependent indices in Fig.6 and Fig.7. It is shown the interval-type error bound appears to be increasing with respect to the interval-size. Moreover, the interval-type FDBT outperforms FGBT and FIBT on both the in-band approximation performance and the error bound for the cases that $\varpi_l<1.5$.

   \begin{figure}[ht!]
    \centering
      \includegraphics[scale=0.55]{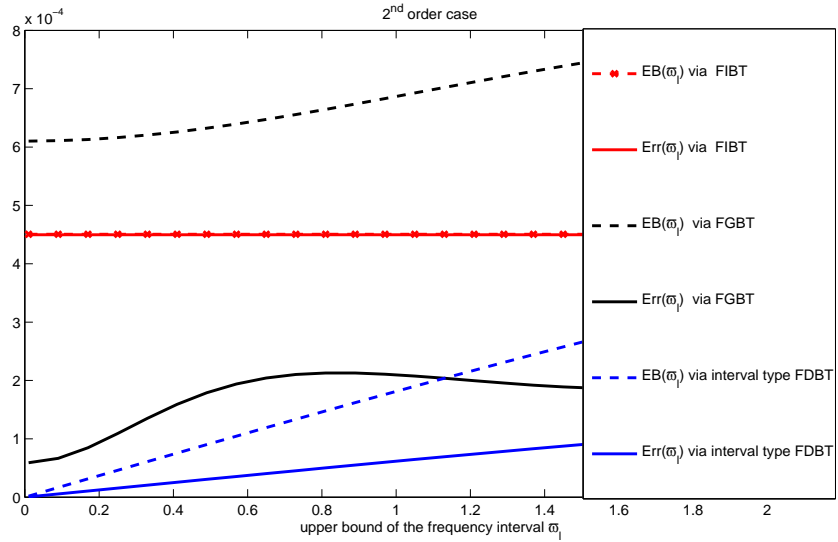}
      \vspace{-0.0cm}
     \caption{Curves of maximum error and error bound ($1^{st}$ order reduced model). $Err(\varpi_l)$: represents the maximum approximation error over frequency interval $[-\varpi_l,\varpi_l]$, i.e. $Err(\varpi_l)=\sigma_{max}(G(\jmath\omega)-G_r(\jmath\omega)), \forall \omega \in [-\varpi_l,\varpi_1]$, where $G_r(\jmath\omega)$ denotes the reduced model  generated by specified method. $Err(\varpi_l)$: represents the interval-type error bound for interval-type FDBT. }
       \end{figure}

   \begin{figure}[ht!]
    \centering
      \includegraphics[scale=0.55]{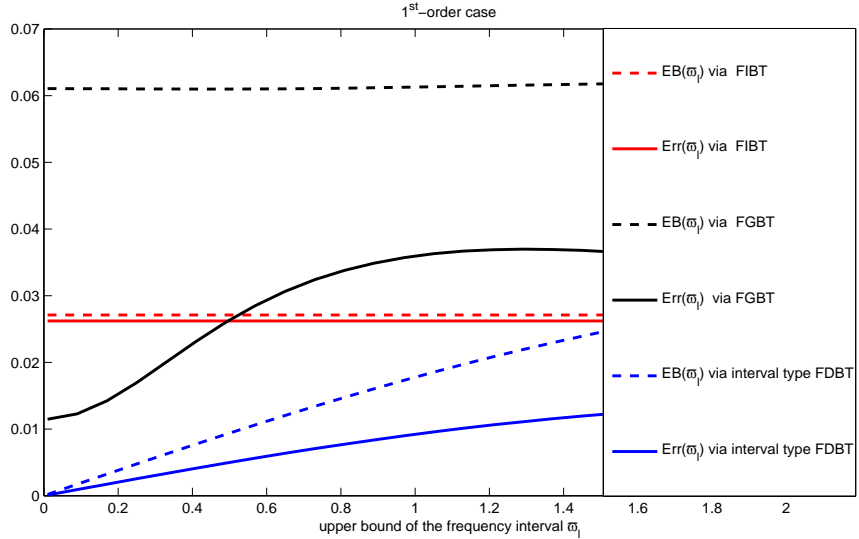}
      \vspace{-0.0cm}
          \caption{Curves of maximum error and error bound ($2^{nd}$ order reduced model). $Err(\varpi_l)$: represents the maximum approximation error over frequency interval $[-\varpi_l,\varpi_l]$, i.e. $Err(\varpi_l)=\sigma_{max}(G(\jmath\omega)-G_r(\jmath\omega)), \forall \omega \in [-\varpi_l,\varpi_1]$, where $G_r(\jmath\omega)$ denotes the reduced model  generated by specified method. $Err(\varpi_l)$: represents the interval-type error bound for interval-type FDBT. }
       \end{figure}

\noindent As referred to in Remark 4, the interval-type FDBT always provides small error bound as long as the size of frequency interval is small enough. To show this, a randomization experiment was carried out. We randomly generate 100 stable systems with order 4. (The off-diagonal elements of matrix $A$ and each element of the matrices $B, C, D$ are obtained with a zero mean and unitary variance normal distribution, the diagonal element of matrix $A$ are obtained with mean -5.5 and variance 4.5). To compare the average performance between FGBT and interval-type FDBT, several indices are defined in Table \uppercase\expandafter{\romannumeral2}.

\begin{table}[!hbp]
\centering
\caption{Indices used to compare the approximation error and error bound generated by different methods}
\vspace{-0.1cm}
\begin{tabular}{l|l}
\Xhline{0.8pt}
   Indexes                     &  computation formula                             \\ [2mm] \Xhline{0.8pt}
   Err(${\varpi _l},r$, FDBT)   & $ \frac{1}{L}\sum\limits_{l = 1}^L {\frac{{{\sigma _{\max }}(G^l(\jmath\omega ) - G_{Dr}^{l}(\jmath\omega )), {\kern 6pt} \omega  \in [ - {\varpi _l}, + {\varpi _l}] }}{{{\sigma _{\max }}(G^l(\jmath \omega ) - {G_{Ir}}^{l}(\jmath\omega )) {\kern 6pt}\omega  \in [ - {\varpi _l}, + {\varpi _l}] }}}$
                                                                    \\ [2mm] \Xhline{0.2pt}
   Err(${\varpi _l},r$, FGBT)   &   $\frac{1}{L}\sum\limits_{l = 1}^L {\frac{{{\sigma _{\max }}(G^l(\jmath\omega ) - G_{Gr}^{l}(\jmath\omega )), {\kern 6pt} \omega  \in [ - {\varpi _l}, + {\varpi _l}]}}{{{\sigma _{\max }}(G^l(\jmath\omega ) - G_{Ir}^{l}(\jmath\omega )), {\kern 6pt} \omega  \in [ - {\varpi _l}, + {\varpi _l}]}} } $
                                    \\ [2mm] \Xhline{0.2pt}
   Eb(${\varpi _l},r$, FDBT)   & $ \frac{1}{L}\sum\limits_{l = 1}^L {\frac{{ upper {\kern 4pt} bound {\kern 4pt} of  {\kern 4pt} \left({\sigma _{\max }}(G^l(\jmath\omega ) - G_{Dr}^{l}(\jmath\omega )), {\kern 6pt} \omega  \in [ - {\varpi _l}, + {\varpi _l}] \right) }}{{upper {\kern 4pt} bound {\kern 4pt} of  {\kern 4pt} \left( {\sigma _{\max }}(G^l(\jmath\omega ) - {G_{Ir}}^{l}(\jmath\omega )) {\kern 6pt}\omega  \in [ - \infty, +\infty]\right) }}}$
                                                                     \\ [2mm] \Xhline{0.2pt}
   Eb(${\varpi _l},r$, FGBT)   &   $\frac{1}{L}\sum\limits_{l = 1}^L {\frac{{ upper {\kern 4pt} bound {\kern 4pt} of  {\kern 4pt} \left( {\sigma _{\max }}(G^l(\jmath\omega ) - G_{Gr}^{l}(\jmath\omega )), {\kern 6pt} \omega  \in [ -\infty, +\infty] \right)}}{{ upper {\kern 4pt} bound {\kern 4pt} of  {\kern 4pt} \left( {\sigma _{\max }}(G^l(\jmath\omega ) - G_{Ir}^{l}(\jmath\omega )), {\kern 6pt} \omega  \in [ -\infty, +\infty] \right)}} } $
                                    \\ [2mm] \Xhline{0.2pt}
  \end{tabular}
\end{table}

\noindent In Table \uppercase\expandafter{\romannumeral2}, $\varpi_l$ represents the upper bound of the symmetrical frequency interval, $r$ is the order of reduced model,  $G_{Dr}^l(\jmath\omega), G_{Sr}^l(\jmath\omega),G_{Gr}^l(\jmath\omega),G_{Ir}^l(\jmath\omega)$ represent the reduced models of order $r$ generated by interval-type FDBT, SPA, FGBT and the classic FIBT for the $l^{th}$ random model, respectively. Fig. 8 and Fig. 9 display the experiment results on the these indices.

    \begin{figure}[ht!]
    \centering
      \includegraphics[scale=0.52]{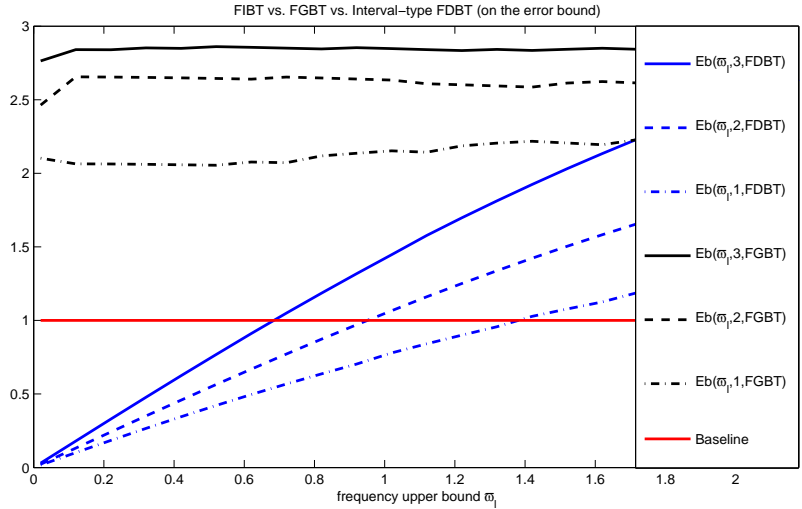}
      \caption{Randomized experiment results on actual error}
       \end{figure}

\noindent Fig. 8 validated that the interval-type error bound provided by interval-type FDBT generally is smaller than the EF-type error bound generated by FIBT and FDBT for the cases that the interval-size is small enough (about $\varpi<1$ in this experiment). Although the advantage on the error bound is restricted for small interval-size cases, it is suggested to take the interval-type FDBT as a feasible option even for medium interval-size cases. According to our experiment, the interval-type FDBT generally also gives rise to better in-band approximation performance than FIBT and FGBT for medium interval-size cases (see Fig. 9 for details). \\

  \begin{figure}[ht!]
    \centering
      \includegraphics[scale=0.52]{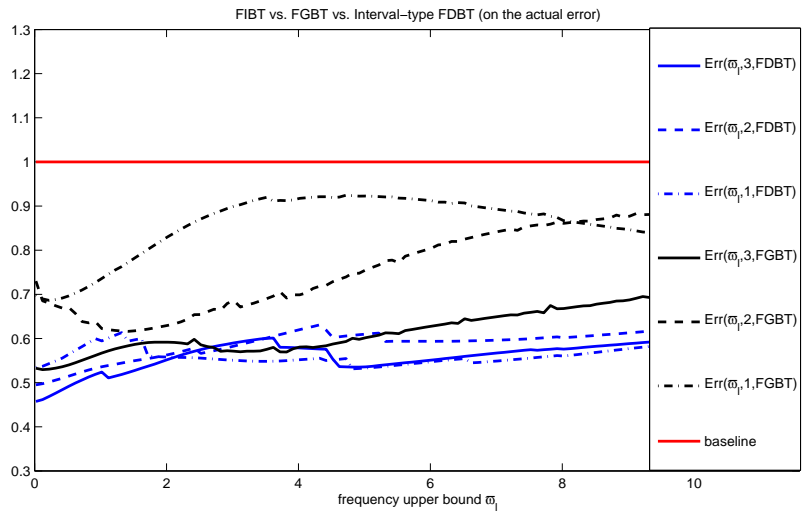}
        \caption{Randomized experiment results on error bound}
       \end{figure}
   \end{example}

\begin{example} Lets consider the $201^{th}$ order RLC ladder circuit example provided by \cite{BT2_Survey} \cite{Ladder_example}. As has been pointed out in \cite{Ladder_example}, approximating the ladder circuit is quite difficult in the framework of balancing related model order reduction approaches since neither the Hankel nor the singular values decay to any extent. In particular, its dynamic behavior over low frequency ranges is too complex to be well approximated due to the special distribution of its poles and zeros. Here we are interested to approximate this circuit in the following cases:\\
\noindent \textbf{Case I}:\; the frequency of input signal belongs to a unknown neighborhood of dominating operating frequency point ($\varpi=0$). \\
\noindent \textbf{Case II}:  the frequency of input signal is known to be within the interval ($\omega \in [-0.5, +0.5]$).\\
\noindent At first, lets consider the case I and apply FIBT and generalized SPA to build reduced models. The frequency response of full model and reduced model of order $181$ are shown in Fig. 10.

   \begin{figure}[ht!]
    \centering
      \includegraphics[scale=0.5]{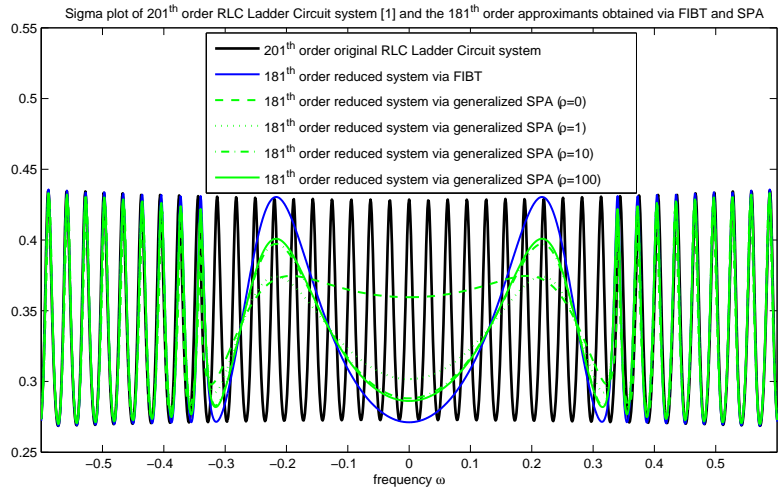}
      \vspace{-0.5cm}
      \caption{Approximating the ladder circuit in Case I via FIBT \& Generalized SPA}
       \end{figure}\noindent As indicated by the visual inspections of the frequency response of the reduced vs. the full system from Fig. 10, the standard FIBT is failed to approximate the dynamic behaviors around $\omega=0$ even the order of reduced model is $181$. Besides, it is surprising and remarkable that the generalized SPA method also failed here. Although the generalized SPA approach generally leads to good approximation performance around $\omega=0$, it is incapable to cope with this example. Now, lets resort to the proposed SF-type FDBT for dealing with the model reduction problem in case I. Our experiment results show that good approximants can be generated via SF-type FDBT as long as the order of reduced system is larger than $50$. The frequency response of the full system and reduced systems in Fig. 11 show a success of SF-type FDBT for this example. Therefore, the SF-type FDBT should be treated as a useful alternative way for solving model reduction problem with low-frequency assumption. In our opinion, it is a non-trivial parallel approach beside the well-known generalized SPA. In addition, the frequency response of reduced model generated by \emph{Pad$\acute{e}$ approximation} (i.e moment-matching at zero) is also included in Fig. 11. It is observed that \emph{Pad$\acute{e}$ approximation} also leads to good  approximation performance, which is both natural and expected since it is an inherent local approximation method. It is interesting that the performance of interval-type FDBT is very similar with \emph{Pad$\acute{e}$ approximation} for this example. The reasons for the similarity is unclear and comparing them is far beyond the scope of this paper. Here we just want to show the possibility that good local approximation performance of the ladder circuit may also be obtained in the balancing-related framework.
   \begin{figure}[ht!]
    \centering
      \includegraphics[scale=0.5]{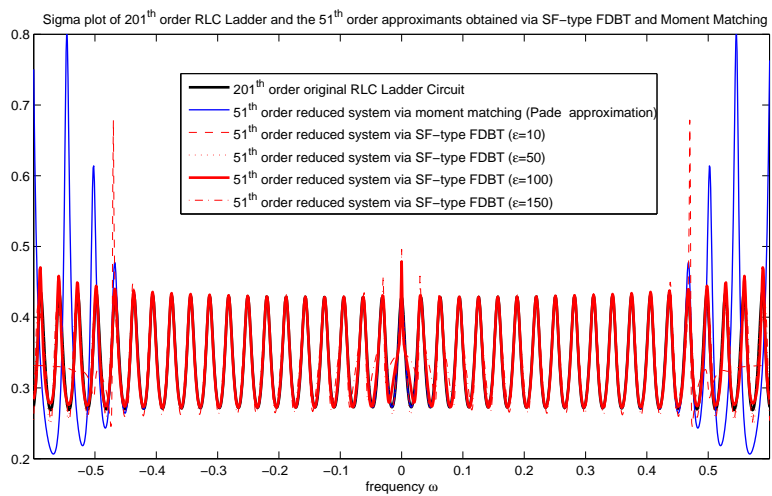}
         \caption{Approximating the ladder circuit in Case I via SF-type FDBT \& Moment matching}
       \end{figure}\\\noindent Finally, lets consider the stated model reduction problem in case II and apply the interval-type FDBT and FGBT \cite{BT2_Survey} to build reduced model. Fig. 12 shows the frequency response of full model and reduced models of order $61$ and $51$.  The results show that only the interval-type FDBT leads to satisfactory in-band approximation performance.

   \begin{figure}[ht!]
    \centering
      \includegraphics[scale=0.5]{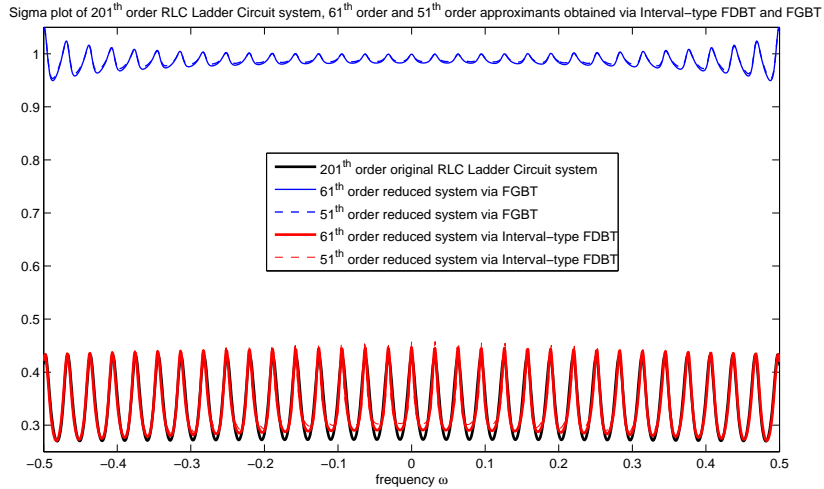}
            \caption{Approximating the ladder circuit in Case II via Interval-type FDBT \& FGBT}
       \end{figure}
\end{example}

\section{Conclusions and Future Work}
This paper revisited model order reduction over limited frequency intervals in the framework of balanced truncation. From a new perspective that establishing  frequency-dependent type error bound instead of the existing frequency-independent type error bound, we developed SF-type and interval-type frequency-dependent  balanced truncation methods to cope with the partially pre-known frequency interval cases and the completely pre-known frequency interval cases, respectively. Moreover, SF-type and interval-type error bound have been established in the first time. Examples have been illustrated to verify the efficiency and advantage of the proposed methods. Future work will focus on developing frequency-dependent balanced truncation algorithms in other forms to get a sharper frequency-dependent error bound.

%








\end{document}